\documentclass{article}
\usepackage{fullpage}

\usepackage{microtype}
\usepackage{graphicx}
\usepackage{subfigure}
\usepackage{booktabs} 

\usepackage{hyperref}

\usepackage[authoryear,sort]{natbib}

\usepackage{amsmath,amssymb,mathtools}
\usepackage{bm}
\usepackage{amsthm}
\usepackage{graphicx,color,wrapfig,epsfig,epstopdf}
\usepackage{enumitem}
\usepackage{algorithm,algorithmic}

\def\x{\bm{x}}
\def\g{\bm{g}}
\def\v{\bm{v}}
\def\bdelta{\bm{\delta}}
\def\m{\bm{m}}
\def\z{\bm{z}}

\def\DS{\texttt{APMSqueeze}}
\def\Cref{\ref}

\def\OA{\texttt{APMSqueeze}}

\newcommand\numberthis{\addtocounter{equation}{1}\tag{\theequation}}

\newtheorem{theorem}{Theorem}
\newtheorem{lemma}{Lemma}
\newtheorem{corollary}{Corollary}
\newtheorem{assumption}{Assumption}

\title{$\OA$: A Communication Efficient Adam-Preconditioned Momentum SGD Algorithm}
\usepackage{authblk}
\author[1,2]{Hanlin Tang}
\author[3]{Shaoduo Gan}
\author[1]{Samyam Rajbhandari}
\author[2]{Xiangru Lian}
\author[2]{Ji Liu}
\author[1]{Yuxiong He}
\author[3]{Ce Zhang}
\affil[1]{Microsoft}
\affil[2]{Department of Computer Science, University of Rochester}
\affil[3]{Department of Computer Science, ETH Zurich}

\begin{document}
\maketitle

\begin{abstract}
  Adam is the important optimization algorithm to guarantee efficiency and accuracy for training many important tasks such as BERT and ImageNet. However, Adam is generally not compatible with information (gradient) compression technology. Therefore, the communication usually becomes the bottleneck for parallelizing Adam. In this paper, we propose a communication efficient {\bf A}DAM {\bf p}reconditioned {\bf M}omentum SGD algorithm-- named \OA-- through an error compensated method compressing gradients. The proposed algorithm achieves a similar convergence efficiency to Adam in term of epochs, but significantly reduces the running time per epoch.
  In terms of end-to-end performance (including the full-precision pre-condition step), \OA
  is able to provide {sometimes by up to $2-10\times$ speed-up
  depending on network bandwidth.} We also conduct theoretical analysis on the convergence and efficiency.
\end{abstract}

\section{Introduction}

Modern advancement of machine learning is
heavily driven by the advancement of
computational power and techniques. Nowadays,
it is
not unusual that a single model requires
hundreds of computational devices such
as GPUs. As a result, scaling up training
algorithms in the distributed setting has
attracted intensive interests over the years.
Example techniques include
quantization~\citep{pmlr-v70-zhang17e,Wangni2018-ux},
decentralization~\citep{Lian2017-ni, Koloskova*2020Decentralized, NIPS2018_8028},
asynchronous communication~\citep{DBLP:journals/corr/ZhengMWCYML16, NIPS2015_6031}.

However, one gap exists in the current research
landscape --- although most distributed training
theory and analysis are developed for the
vanilla version of stochastic gradient descent (SGD),
in reality many state-of-the-art models have to be
trained using more complicated variant. For example,
to train state-of-the-art models such as
BERT, one has to resort to the Adam
optimizer, since training it with vanilla/momentum SGD has
been shown to be less effective.
However, it is not clear how these more advanced
optimizer can be scaled up --- and as we will see,
directly applying techniques researchers
developed for SGD often fails to work well for these
optimizers (See Section~\ref{resnet}). In this paper, we ask, {\em How can we
scale up sophisticated optimizers beyond
SGD?}

In this paper, we focus on one specific
optimizer, i.e., Adam, and one specific
optimization technique, i.e., communication
compression. We first analyze the limitation of
directly applying exsiting technique to
Adam. We then propose a new algorithm, $\OA$, which,
instead of applying communication
compression to Adam, uses Adam to ``pre-condition''
a communication compressed momentum SGD algoirthm.
This algorithm is powerful and it matches
Adam's result in training demanding
ML models such as BERT-Large, while communicates
$16$-$32\times$ less data per epoch.
We provide theoretical analysis
on
communication compressed momentum SGD,
which is the core component of $\OA$.
We then conduct
extensive experiments (up to BERT Large on 
128 GPUs) and show that, under different 
network conditions, $\OA$ is able to 
provide up to {one order of magnitude 
speed-up} on per-iteration runtime
({include the full-precision pre-condition step})
, while
maintaining the same empirical convergence
behavior.

{\bf (Contributions)}
We make the following
contributions.
\begin{itemize}
\item We propose a new algorithm, \OA, a
communication efficient, momentum SGD algorithm
pre-conditioned with a few epochs
of a distributed Adam optimizer. We present
novel, non-trivial analysis on the convergence of
the algorithm, and show that the compressed algorithm admits the same asymptotic convergence rate as the uncompressed one.
\item We conduct experiments on large scale
ML tasks that are currently challenging for
SGD to train. We show that on both BERT-Base
and BERT-Large, our algorithm is able to achieve
the same convergence behaviour and final
accuracy as Adam, with as large as
32$\times$ communication compression. In many cases,
this reduces the training time by up to {$4$-$8\times$.}
({include the full-precision pre-condition step})
To our best knowledge, this is the first distributed
learning algorithm with communication compression
that can train a model as demanding as BERT.
\end{itemize}

\paragraph*{Problem Setting} In this paper, we focus on
the following optimization task and rely on the following
notions and definitions.
\begin{equation}
\min_{\bm{x}}\quad f(\bm{x}) = {1\over n} \sum_{i=1}^n \mathbb{E}_{\bm{\bm{\zeta}}\sim\mathcal{D}_i}F(\bm{x}; \bm{\bm{\zeta}}),\label{eq:main}
\end{equation}
\paragraph{Notations and definitions}
Throughout this paper, we use the following notations:
\begin{itemize}
\item $\nabla f(\cdot)$ denotes the gradient of a function $f$.
\item $f^{*}$ denotes the optimal value of the minimization problem \eqref{eq:main}.
\item $f_i(x) := \mathbb{E}_{\xi\sim\mathcal{D}_i}F_{i}(x; \xi)$.
\item $\|\cdot\|$ denotes the $\ell_2$ norm for vectors and the spectral norm for matrices.
\item $\|X\|_A:=\text{Tr}(X^{\top}AX)$.
\item $\bm{C}_{\omega}(\cdot)$ denotes the randomized compressing operator, where $\omega$ denotes the random variable. One example is the randomized quantization operator, for example, $\bm{C}_{\omega}(0.7) = 1$ with probability $0.7$ and $\bm{C}_{\omega}(0.7) = 0$ with probability $0.3$. It is also worth noting that this notion $\bm{C}_{\omega}(\cdot)$ also covers the deterministic scenario, for example, $\bm{C}_{\omega}(\cdot)$ is a one bit compression operator.
\item { $\sqrt{\cdot}$ denotes the square root of the augment. In this paper, we abuse this notation a little bit. If the augment is a vector, then it returns a vector taking the element-wise square root.}
\item { $\oslash$ denotes the element-wise divide operator, that is, the $i$th element of $\bm{m} \oslash \bm{v}$ is $\bm{m}_i / \bm{v}_i$.}
\item $\odot$ denotes the element-wise multiply operator, that is, the $i$th element of $\bm{m} \odot \bm{v}$ is $\bm{m}_i * \bm{v}_i$.
\end{itemize}

\section{Related Work}
\paragraph{Communication-Efficient Distributed Learning:}
To further reduce the communication overhead, one promising direction is to compress the variables that are sent between different workers ~\citep{NIPS2019_8694,NIPS2019_9473}. Previous work has applied a
range of techniques such as quantizaiton,
sparsification, and sketching
~\citep{Alistarh2017-yh,Agarwal2018-hg,Spring2019-ep,Ye2018-mf}.
The compression is mostly assumed to be unbiased ~\citep{Wangni2018-ux,pmlr-v80-shen18a,pmlr-v70-zhang17e,NIPS2017_6749,NIPS2018_7519}.
A general theoretical analysis of centralized compressed parallel SGD can be found in ~\citet{Alistarh2017-yh}. Beyond this, some biased compressing methods are also proposed and proven to be quite efficient in reducing the communication cost. One example is the 1-Bit SGD ~\citep{1-bitexp}, which compresses the entries in gradient vector into $\pm 1$ depends on its sign. The theoretical guarantee of this method is given in ~\citet{Bernstein:2018aa}.

\paragraph{Error-Compensated Compression:}
The idea of using error compensation for compression is proposed in ~\citet{1-bitexp}, where they find that by using error compensation the training could still achieves a very good speed even using $1$-bit compression. Recent study indicates that this strategy admits the same asymptotic convergence rate as the uncompressed one for single-pass~\citep{martinmemory} or double-pass~\citep{tangdouble} parameter server communication, which means that the influence of compression is trivial. More importantly, by using error compensation, it has been proved that we can use almost any compression methods~\citep{tangdouble}, whereas naive compression could only converge when the compression is unbiased (the expectation of the compressed tensor is the same as the original). Due to the promising efficiency of this method, error compensation has been applied into many related area ~\citep{NIPS2019_9321,9051706,NIPS2019_8694,8884924,NIPS2019_9473,NIPS2019_8598,NIPS2019_9610,NIPS2019_9571} in order to reduce the communication cost. 
\paragraph{Adam:} Adam~\citep{Kingma2015AdamAM} has shown
promising speed for many deep learning tasks, and also admits a very good robustness to the choice of the hyper-parameters, such as learning rate. 
It can be viewed as an adaptive method that scales the learning rate with the magnitude of the gradients on each coordinate when running SGD. Beyond Adam,  many other strategies that that shares the same idea of changing learning rate dynamically was studied. For example,  \citet{JMLR:v12:duchi11a} (Adagrad) and \citep{rmsprop} (RESprop),  use the gradient, instead of momentum, for updating the parameters;  Adadelta~\citep{DBLP:journals/corr/abs-1212-5701} changes the variance term of Adam into a non-decreasing updating rule; \citet{luo2018adaptive} proposed AdaBound that gives both upper bound and lower bound for the variance term.


\begin{algorithm}[t]\caption{$\OA$}\label{alg:global}
\begin{algorithmic}[1]
\STATE {\bfseries Initialize}: $\x_0$, learning rate $\gamma_t$, averaging rate $\eta$, initial error $\bdelta = \boldsymbol{0}$, $\m_0 = \boldsymbol{0}$, $\v_0 = \boldsymbol{0}$, number of total iterations $T$,  warm-up steps $T_{w}$.

\STATE Running the original Adam for $T_{w}$ steps.
\FOR {$t=T_w,\ldots,T$}

\STATE \textbf{(On $i$-th node)}
\STATE  Randomly sample $\bm{\xi}_t^{(i)}$  and compute local stochastic gradient $\g_t^{(i)} := \nabla F_i(\x_t^{(i)}, \bm{\xi}_t^{(i)})$.

\STATE Update the local momentum variable $\m_{t-1}$ according to
$
\m_t^{(i)} =  \beta_1\bm{m}_{t-1} + (1 - \beta_1)\g_t^{(i)}.
$
\STATE  Divide $\m_t^{(i)}$ into $n$ chunks. Compress its $k$-th chunk (denote as $\m_t^{(i,k)}$) into $\bm{C}_\omega\left[\m_t^{(i,k)} + \bdelta_t^{(i,k)}\right]$, and update the compression error by $\bdelta_t^{(i,k)} = \m_t^{(i,k)} + \bdelta_t^{(i,k)} - \bm{C}_\omega[\m_t^{(i,k)}]$.
\STATE Send the  $\bm{C}_\omega\left[\m_t^{(i,k)}\right]$ to worker $k$. Receive the $i$-th chunk of $\bm{C}_\omega\left[\m_t^{(j,i)}\right]$ from all other workers with $j\in\{1,\cdots,n\}$.
\STATE Take the average over all $\bm{C}_\omega\left[\m_t^{(j,i)}\right]$ it receives and compress it into
$
\bm{C}_\omega\left[\overline{\m}_t^{(:,i)}\right] =\bm{C}_\omega\left[ \frac{1}{n}\sum_{j=1}^n \bm{C}_\omega\left[\m_t^{(j,i)}\right] + \overline{\bdelta}_{t-1}^{(:,i)}\right],
$
 and update the compression error accordingly by $\overline{\bdelta}_t^{(:,i)} = \frac{1}{n}\sum_{j=1}^n \bm{C}_\omega\left[\m_t^{(j,i)}\right] + \overline{\bdelta}_{t-1}^{(:,i)} - \bm{C}_\omega\left[\overline{\m}_t^{(:,i)}\right]$.
 \STATE Send $\bm{C}_\omega\left[\overline{\m}_t^{(:,i)}\right]$ to all the workers, then replace the $k$-th chunk of the original momentum $\m_t$ with $\bm{C}_\omega\left[\overline{\m}_t^{(:,k)}\right]$ after receiving it.
\STATE Update local model $\x_{t+1} = \x_t - \gamma_t \m_t\oslash\sqrt{\bm{v}_{_{\tiny T_w}}}$.
\ENDFOR
\STATE {\bfseries Output}: $\x$.
\end{algorithmic}\label{alg:de_ec}
\end{algorithm}

\section{$\OA$ Algorithm}
In this section, we will introduce $\OA$ in detail. We start with some background for error compensated compression and Adam, then we will give full description of $\OA$.

\subsection{Error Compensated Compression}
One standard  way to reduce the communication overhead for SGD is to compress the gradient before sending, which can be expressed as
\begin{align*}
\x \gets \x - \gamma C_{\omega}[\g].
\end{align*}
where $C_{\omega}[\cdot]$ {\footnote{$C_{\omega}[\cdot]$ could also include randomness.}} is the compress operator. The problem with this straightforward strategy is that the compression would slow down the training speed or even make the training diverge, because the many information would be lost after compression. Recent studies{~\citep{martinmemory,tangdouble}} shows that actually the information lost can be stored and got recovered in the next step, which is called error compensated compression.
The idea is to store compression error as $\bm{\delta}$, and send $C_{\omega}\left[\bm{g}+\bm{\delta}\right]$, where we update
$\bm{\delta}$ by using the following recursion at each time step
\[
\bm{\delta} \gets  \bm{g}+\bm{\delta} -C_{\omega}\left[\bm{g}+\bm{\delta}\right] .
\]

\subsection{Original Adam} Unlike SGD, instead of applying the gradients $\g$ directly to update the model $\x$, Adam uses two auxiliary variables $\m$ and $\v$ for the update. The mathematical updating rule of  original Adam can be summarized as:
\begin{align*}
\bm{m}_{t+1} = \beta_1\bm{m}_t + (1-\beta_1)\g_t,&\quad
\bm{v}_{t+1} =  \beta_2\bm{v}_t + (1-\beta_2)(\bm{g}_t)^2,\numberthis\label{alg:v}\\
\hat{\m} = \frac{\m}{1-\beta_1^t},&\quad
\hat{\bm{v}} = \frac{\bm{v}}{1-\beta_2^t},\\
\bm{\x}_{t+1} =  \x_t -& \gamma\hat{\bm{m}}\oslash\left(\sqrt{\hat{\bm{v}}} + \eta\right),
\end{align*}
Here $\x_t$ is the model at $t$-iteration, $\g_t = \nabla F(\x_t;\zeta_t)$ is the stochastic gradient $t$-iteration, $\gamma$ is the learning rate, $\eta$ is a constant, so as $\beta_1$ and $\beta_2$.  $\bm{m}$, $\bm{v}$,$\hat{\m}$, $\hat{\bm{v}}$ are auxiliary variables.

{\color{black}As we can see, Adam is non-linearly dependent to the gradient, and this non-linearity would lead to some intrinsic problems to the combination of Adam and error-compensation (see Supplement for more details). This leads us to $\OA$, which is capable of achieving almost the same convergence rate and can be easily combined with error compensation. }

\begin{figure}[t]
\centering
\subfigure[\scriptsize \textbf{Gather step}: Each worker sends its $i$-th chunk to worker $i$.]{
\begin{minipage}[t]{0.3\linewidth}
\centering
\includegraphics[width=1\textwidth]{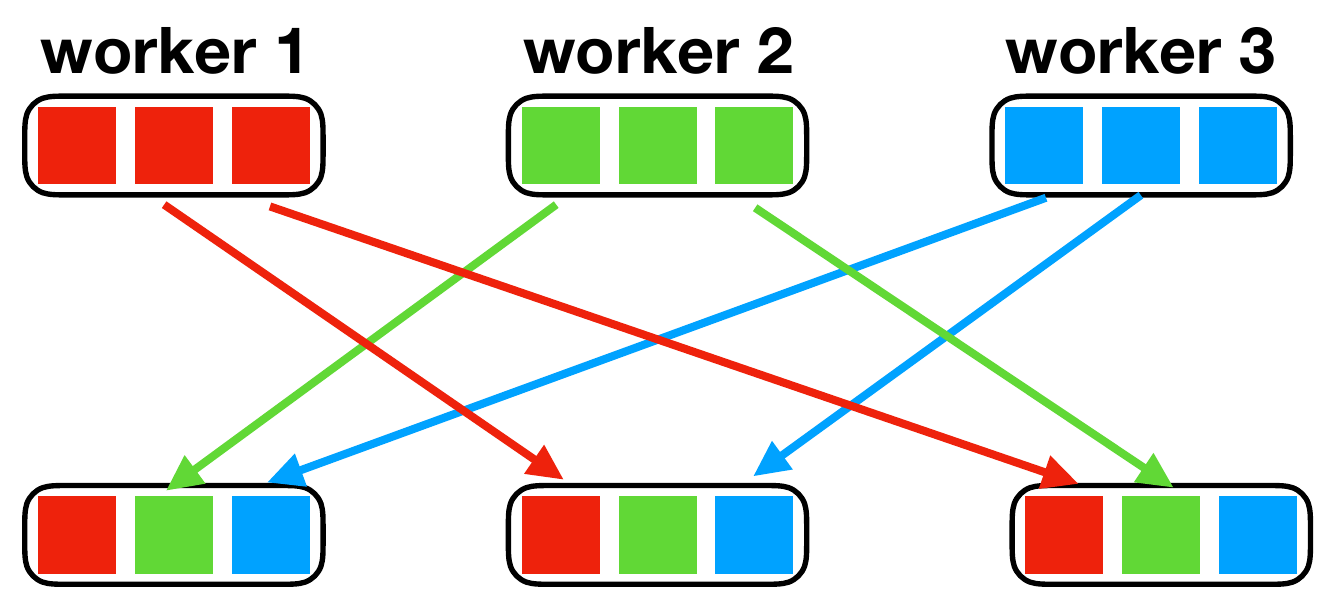}
\end{minipage}
}\quad
\subfigure[\scriptsize \textbf{Average step}: Each worker averages all chunks it receives.]{
\begin{minipage}[t]{0.3\linewidth}
\centering
\includegraphics[width=1\textwidth]{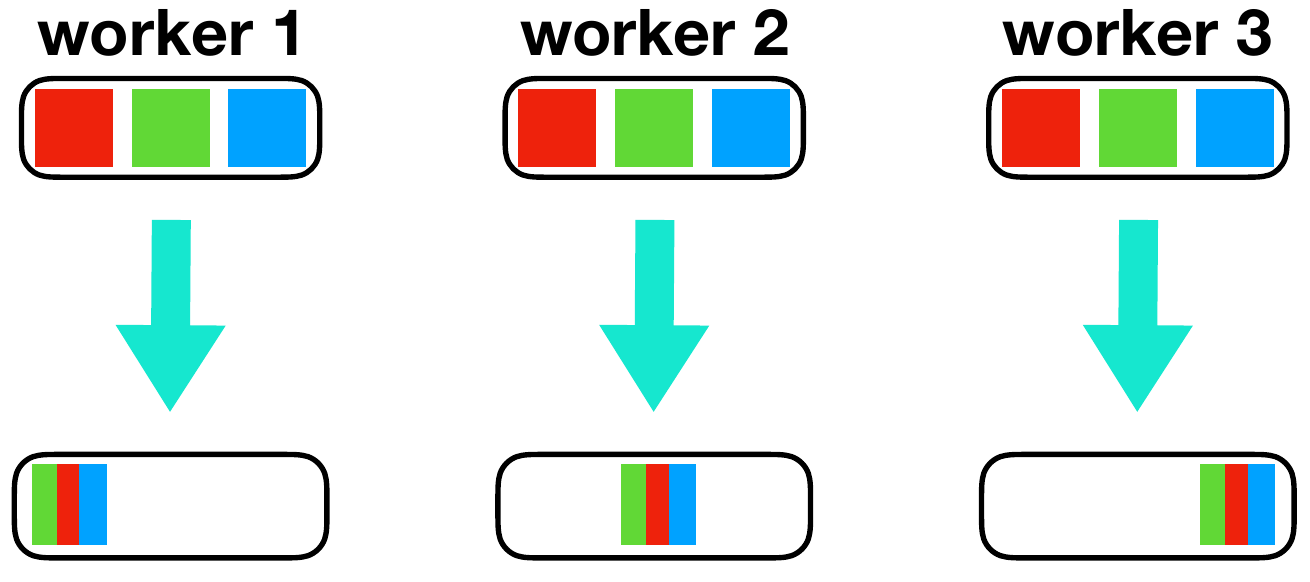}
\end{minipage}%
}\quad
\subfigure[\scriptsize \textbf{Scatter step}: Each worker receives the $i$-th chunk from worker $i$.]{
\begin{minipage}[t]{0.3\linewidth}
\centering
\includegraphics[width=1\textwidth]{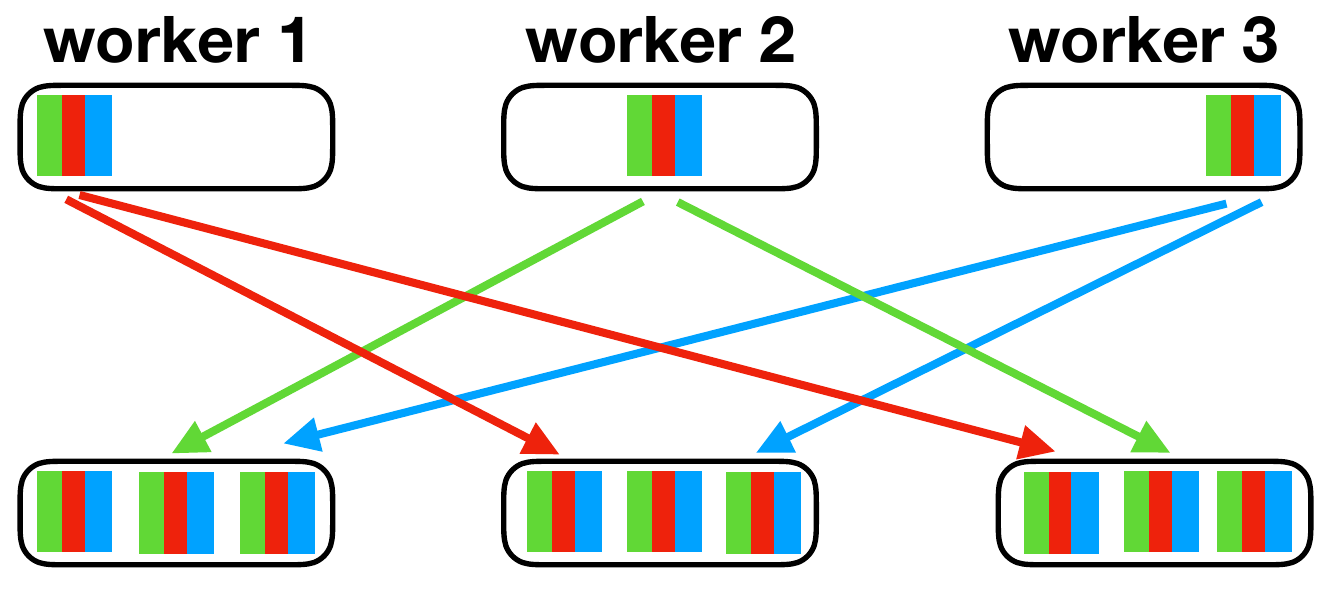}
\end{minipage}%
}%
\centering
\caption{Pipeline for Gather-Scatter AllReduce}\label{allreduce}
\end{figure}

\subsection{$\OA$}\label{alg:description}
In $\OA$, we only use Adam for a few epochs for warm-up, and after the warm-up stage, we would stop updating $\v$. The detailed description is stated below.

Consider that there are $n$ workers in the network. In order to fully utilize the bandwidth of the network, we use the Gather-Scatter AllReduce ~\citep{pmlr-v97-yu19f} (See Figure~\ref{allreduce}) prototype for the realization of the parameter-server parallelism.

Below are the steps of $\OA$:
\begin{enumerate}
\item \textbf{Local Computation:} Update the error-compensated momentum by $\m_t^{(i)} = \beta_1\bm{m}_{t-1} + (1 - \beta_1)\g_t^{(i)} + \bdelta_t^{(i)}$, where $\g_t^{(i)}$ is the stochastic gradient, $\beta_1$ is the a scalar.

\item \textbf{Local Compression:} Divide $\m_t^{(i)}$ into $n$ chunks, denote the $k$-th chunk of $\m_t^{(i)}$ as $\m_t^{(i,k)}$. Compress each chunk into $\bm{C}_\omega\left[\m_t^{(i,k)}\right]$, where $\bm{C}_\omega[\cdot]$ is the compression operation, and update the compression error by $\bdelta_t^{(i,k)} = \m_t^{(i,k)} - \bm{C}_\omega[\m_t^{(i,k)}]$.

\item \textbf{Scatter:} Send the $k$-th chunk $\bm{C}_\omega\left[\m_t^{(i,k)}\right]$ to worker $k$, and receive the $i$-th chunk of $\bm{C}_\omega\left[\m_t^{(k,i)}\right]$ from worker $k$ for all $k\in\{1,\cdots,n\}$.

\item\textbf{Local Average:}  Updating the averaged momentum with $\overline{\m}_t^{(:,i)} = \frac{1}{n}\sum_{k=1}^n \bm{C}_\omega\left[\m_t^{(k,i)}\right] + \overline{\bdelta}_{t-1}^{(:,i)}$. Recompress $\overline{\m}_t^{(:,i)}$ into $\bm{C}_\omega\left[\overline{\m}_t^{(:,i)}\right]$ and update the compression error by $\overline{\bdelta}_t^{(:,i)} = \frac{1}{n}\sum_{j=1}^n \bm{C}_\omega\left[\m_t^{(j,i)}\right] + \overline{\bdelta}_{t-1}^{(:,i)} - \bm{C}_\omega\left[\overline{\m}_t^{(:,i)}\right]$.

\item\textbf{Gather:} Send $\bm{C}_\omega\left[\overline{\m}_t^{(:,i)}\right]$ to all the workers. After receiving $\bm{C}_\omega\left[\overline{\m}_t^{(:,k)}\right]$ from all other workers, replace the $k$-th chunk of the original momentum with $\m_t^{(i,k)} = \bm{C}_\omega\left[\overline{\m}_t^{(:,k)}\right]$, which leads to
$
	\m_t = \left(\bm{C}_\omega\left[\overline{\m}_t^{(:,1)}\right],\cdots,\bm{C}_\omega\left[\overline{\m}_t^{(:,n)}\right]  \right)
$
\item \textbf{Model Update:} Update the model according to
$
	\x_{t+1} = \x_t - \gamma_t \m_t\oslash\sqrt{\bm{v}_{_{\tiny T_w}}},
$ where $\gamma_t$ is the learning rate and $v_{T_w}$ is the variance term computed at the end of the warmup step.
\end{enumerate}
Finally, the proposed $\OA$ algorithm is summarized in Algorithm \ref{alg:de_ec}.

\section{Theoretical Analysis}
In this section, we first introduce some assumptions that is necessary, then we present the theoretical guarantee of the convergence rate for $\OA$.

\begin{assumption}\label{ass:global}
We make the following assumptions:
\begin{enumerate}
\item \textbf{Lipschitzian gradient:} $f(\cdot)$ is assumed to be  with $L$-Lipschitzian gradients, which means
  \begin{align*}
  \|\nabla f(\bm{x}) - \nabla f(\bm{y}) \| \leq L \|\bm{x} - \bm{y} \|,\quad \forall \bm{x},\forall \bm{y},
  \end{align*}
 \item\label{ass:var} \textbf{Bounded variance:}
The variance of the stochastic gradient is bounded
\begin{align*}
\mathbb E_{\zeta\sim\mathcal{D}_i}\|\nabla F(\bm{x};\bm{\zeta}) - \nabla f(\bm{x})\|^2 \leq \sigma^2,\quad\forall \bm{x},\forall i.
\end{align*}
\item \textbf{Bounded magnitude  of error for $C_{\omega}[\cdot]$:}
The magnitude of worker's local errors $\bm{\delta}_t^{(i,k)}$  and the server's global error $\overline{\bdelta}_t^{(:,i)}$, are assumed to be bounded by a constant $\epsilon$
\begin{align*}
\sum_{k=1}^n\mathbb E_{\omega} \left\|\bm{\delta}_t^{(i,k)}\right\|\leq \frac{\epsilon}{2},\quad\forall t,\forall i,\quad
\sum_{i=1}^n\mathbb E_{\omega}\left\|\overline{\bdelta}_t^{(:,i)}\right\|\leq  \frac{\epsilon}{2},\quad\forall t.
\end{align*}
\end{enumerate}
\end{assumption}

Next we are ready to present the main theorem for $\OA$.
\begin{theorem}\label{theo:global}
 Under Assumption~\ref{ass:global}, for $\OA$, we have the following convergence rate
 \begin{align*}
   &\left(1-\frac{\gamma L}{v_{\min}} - \frac{2\gamma^2 L^2}{(1-\beta)^2v_{\min}^2} \right)\sum_{t=0}^T \mathbb E\|\nabla f(\bm{x}_t)\|^2_{V}\\
    \leq & \frac{2\mathbb E f(\bm{x}_{0}) - 2\mathbb Ef(\bm{x}^*)}{\gamma}   + \frac{6\gamma^2L^2\epsilon^2 T}{(1-\beta)^2v_{\min}^3}  +  \frac{L\gamma \sigma^2T}{nv_{\min}} + \frac{2\gamma^2L^2\sigma^2  T}{n(1-\beta)^2v_{\min}^2},\numberthis\label{main:theo:eq}
\end{align*}
where $V= \text{diag}\left(1/\v_{T_w}^{(1)},1/\v_{T_w}^{(2)},\cdots,1/\v_{T_w}^{(d)}\right)$ is a diagonal matrix spanned by $\v_{T_w}$ and $v_{\min} = \min\{\v_{T_w}^{(1)},\v_{T_w}^{(2)},\cdots,\v_{T_w}^{(d)}\}$ is the mimimum value in $\v_{T_w}$
\end{theorem}

Given the generic result in Theorem~\ref{theo:global}, we obtain the convergence rate for $\OA$ with appropriately chosen the learning rate $\gamma$.

\begin{corollary}\label{coro:global}
Under Assumption~\ref{ass:global}, for $\OA$, choosing
$
\gamma = \frac{1}{4L(v_{\min})^{-1} + \sigma\sqrt{\frac{ T}{n}} + \epsilon^{\frac{2}{3}} T^{\frac{1}{3}}(v_{\min})^{-1} },
$
we have the following convergence rate
\begin{align*}
\frac{1}{Tv_{\min}}\sum_{t=0}^{T-1}\mathbb{E}\|\nabla f(\bm{x}_t)\|^2_V \lesssim \frac{\sigma}{\sqrt{nT}} + \frac{\epsilon^{\frac{2}{3}}}{T^{\frac{2}{3}}} + \frac{1}{ T},
\end{align*}
where we treat $f(\bm{x}_1) - f^*$, $\beta$ and $L$ as constants.
\end{corollary}

This result suggests that
\begin{itemize}
\item ({\bf Comparison to SGD}) DoubleSqueeze essentially admits the same convergence rate as SGD in the sense that both of them admit the asymptotical convergence rate $O(1/\sqrt{T})$;
\item ({\bf Linear Speedup}) The asymptotical convergence rate of $\DS$ is $O(1/\sqrt{nT})$, the same convergence rate as Parallel SGD. It implies that the averaged sample complexity is $O(1/ (n\epsilon^2))$.
\end{itemize}

\section{Experiments}

\begin{table*}
  \caption{Results on GLUE.  BERT-Base (original) and BERT-Large(original) results
are from \citet{bert}; BERT-Base (uncompressed) and BERT-Large (uncompressed) are the results that uses the full-precision BertAdam and the same training parameters with the $\DS $ for training; BERT-Base (compressed) and BERT-Large (compressed) are te results using $\DS .$ The scores are the
median scores over 10 runs. We report the accuracy results for those tasks.}\label{table1}
  \centering
  \small
  \begin{tabular}{cccccccc}
  \hline  
  \textbf{Model}& RTE& MRPC& CoLA & SST-2& QNLI& QQP& MNLI-(m/mm) \\
  \hline  
  BERT-Base (original) & 66.4 & 84.8 & 52.1  & 93.5 & 90.5& 89.2& 84.6/83.4\\
  BERT-Base (uncompressed) & 68.2 & 84.8 & 56.8  & 91.8 & 90.9& 90.9& 83.6/83.5\\
  BERT-Base (compressed) & 69.0& 84.8 & 55.6   & 91.6 & 90.8& 90.9& 83.6/83.9\\
  \hline
  BERT-Large (original) & 70.1& 85.4 & 60.5   & 94.9 & 92.7& 89.3& 86.7/85.9\\
  BERT-Large (uncompressed) & 70.3& 86.0 & 60.3   & 93.1 & 92.2& 91.4& 86.1/86.2\\
  BERT-Large (compressed) & 70.4& 86.1 & 62.0  & 93.8 & 91.9& 91.5& 85.7/85.4\\
  \hline 
  \end{tabular}\vspace{-0.5cm}
\end{table*}


We validate our theory with experiments that compared {$\DS$} with other
 implementations. We evaluate the performance of our algorithm for both BERT-Base ,BERT-Large, and ResNet-18.
We show that the {$\DS$} converges similar to Adam without
compression, but runs much faster than uncompressed
algorithms when bandwidth is limited.

\subsection{Compression Method}
We use the two compression methods described below:
\begin{itemize}
    \item \emph{1-bit compression}: The gradients are quantized into 1-bit
  representation (containing the sign of each element). Accompanying the
  vector, a scaling factor is computed as
  $\frac{\text{magnitude of
      compensated gradient}}{\text{magnitude of quantized gradient}}.$ The
  scaling factor is multiplied onto the quantized gradient whenever the
  quantized gradient is used, so that the recovered gradient has the same
  magnitude of the compensated gradient. This compression could reduce the $97\%$ communication cost of the original for float32 type training and $94\%$ for float16 type training.
  \item \emph{Top-$k$ compression}: We take top $k\%$ elements of the original gradient that is sorted by its absolute magnitude. The communication cost is reduced into $k\%$ of the original.
\end{itemize}
For BERT-Base and BERT-Large, we use $1$-bit compression. For ResNet-18, we use both $1$-bit compression and Top-$k$ compression.

\subsection{BERT Training}
\paragraph{Dataset and models} We benchmark the performance of $\DS $ for both  BERT-Base ($L =12$, $ H = 768$, $ A = 12$ , $110M $ params) and BERT-Large ($ L=24$, $H=1024 $,
$A=16 $, $340M $ params). For pretrain task, the dataset is the same as \citet{bert}, which is a concatenation of Wikipedia and BooksCorpus with $2.5B $
and $ 800M$ words respectively. For fine-tuning task, we use the GLUE benchmark \citep{glue}.
\paragraph{Hardware} For BERT-Base, we use 32  GPUs resident on 2 servers, and for BERT-Large we use 128 GPUs resident on 8 servers.  Each macine has 16 GPUs and each GPU is treated as a single worker. The total batch-size is $4K $ for both tasks.

\paragraph{Training Parameters }
For pre-training, the learning rate would linearly increase to $4\times 10^{-4}$  for warm-up in the first $12.5k $ steps,  and then decays into $0.99$ of the original after every $520$ steps. We set the two parameters in Algorithm~\ref{alg:global} as $\beta_1 = 0.9$ and $\beta_2 = 0.999$. 
 Unlike previous work \citep{bert}, where they use $90\%$ training step for a sequence length of 128 and then increase the sequence length to 512 for both BERT-Base and BERT-Large. In our experiment, for BERT-Base,  $118k$ steps are used  for 128 sequence length training and $22k$ steps use a sequence length of 512; while for BERT-Large, $152k$ steps are used  for 128 sequence length training and $10k$ steps use a sequence length of 512. When using sequence 128, the Adam pre-conditioned step before compression for BERT-Base is $16k$ and  $23k$ for BERT-Large. When using sequence length 512, we use $1.5k$ steps of Adam pre-conditioned steps for both tasks. 
\vspace{-0.3cm}
 \paragraph{Convergence Results}
 In Figure~\ref{fig:bert}, we report the sample-wise convergence result for the pretrain task using sequence length of 128, which consists most of the training steps. We use the BertAdam \citep{bert} optimizer as the uncompressed baseline.
  From those two figures we shall see that after the Adam pre-conditioned stage, the training effciency remains almost the same for the compressed training and the uncompressed one, while the communication is reduced into $6\%$ of the original.
\vspace{-0.3cm}
\paragraph{GLUE Results}
For GLUE we consider perform the single-task training on the dev set.
In fintuning, we serach over the hyperparameter space with batch sizes $\in\{8,16\}$
and learning rates $\in\{1\times 10^{-5},3\times 10^{-5},5\times 10^{-5},8\times 10^{-5}\}$. Other setting are the same as  pre-train task.  We report the median development
set results for each task over 10 random initializations.

Results are presented in Table~\ref{table1}. We compare our results with the uncompressed training baseline that uses the same training parameters with the compressed one for a fair comparison. We shall see that the compressed training could achieve a comparable performance with the uncompressed baseline. We also include the results from the previous work \citet{bert}.

\begin{wrapfigure}{r}{0.35\textwidth}
\vspace{-1em}
\centering
\includegraphics[width=0.35\textwidth]{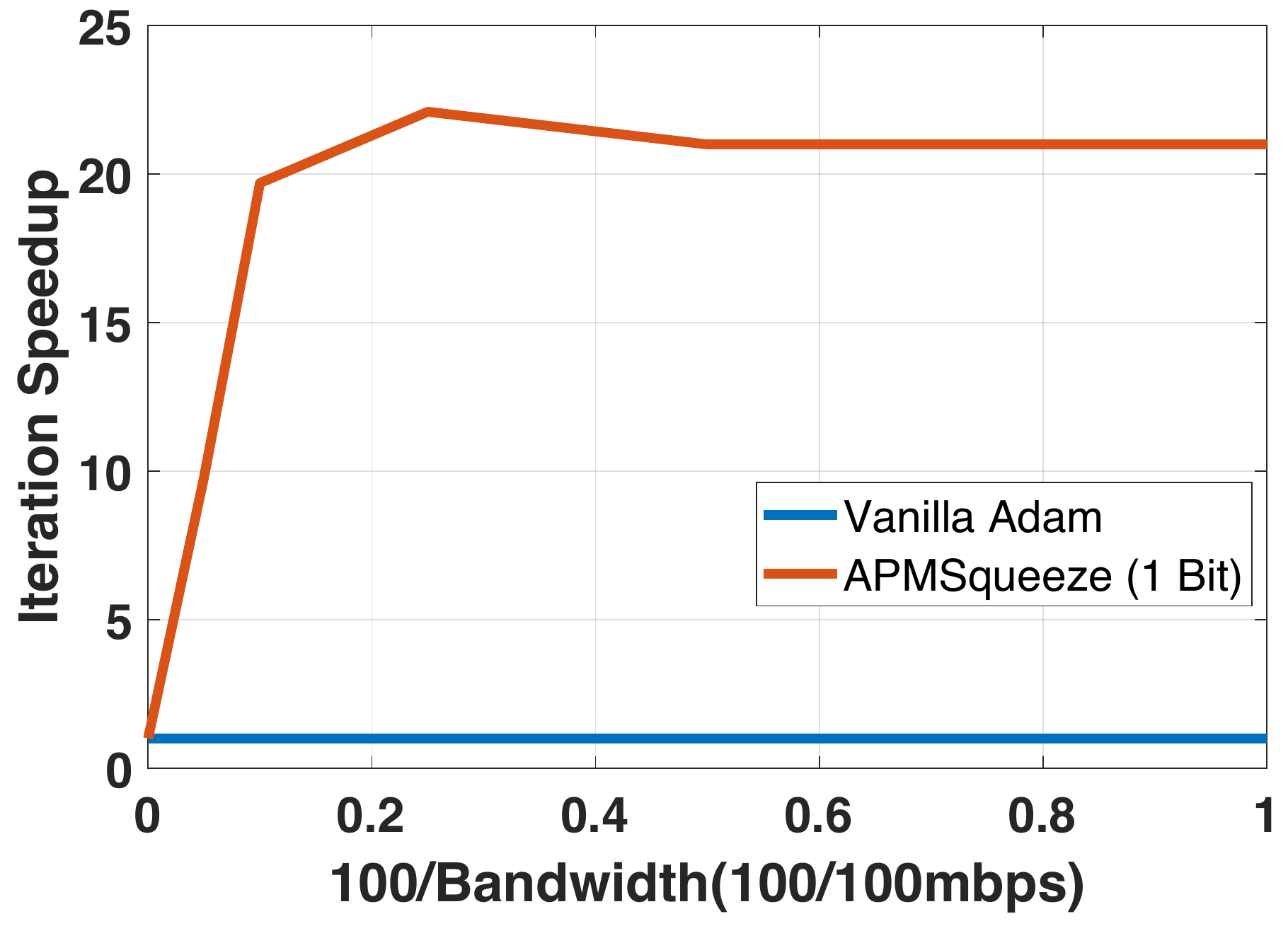}
\caption{Per-iteration speedup for BERT-base under different network bandwidth}\label{fig:bert_speed}\vspace{-0.5cm}
\end{wrapfigure}

\vspace{-0.3cm}
\paragraph{Per-Iteration Speed-up}
Figure~\ref{fig:bert_speed} shows the speed-up performance of our compressed algorithm. We use 64 GPUs and treat each as a separate worker. We train the BERT-Base model from scratch. The communication backend is OpenMPI 4.0.3 compiled with CUDA-aware support. With the proposed algorithm, we compress the communication data from 32bits to 1bit and then report the speed-up of iteration time under different network conditions. Specifically, we use traffic control utility \textit{tc} to shape the bandwidth from 100Gbits to 100Mbits. With the network going slow, the compressed case can achieve a stable speedup by around 22$\times$ over the uncompressed case.
We already reached $10\times$
speed-up for 2Gbits bandwidth
and $3\times$ for 10Gbits bandwidth.

\paragraph{End-to-end Speed-up}
When consider the end-to-end speed-up we
also need to factor in the pre-condition phase, in which we have to run
Adam with full precision. In our experiments, we set the pre-condition
run as the first 15\% of the execution.
When the network is 10Gbits we 
obtain 2$\times$ end-to-end
speed-up, while when the network is
1Gbits, we obtain 7$\times$
end-to-end speed-up.

\begin{figure}[t!]
\centering
\subfigure[BERT-Base 128]{
\begin{minipage}[t]{0.4\linewidth}
\centering
\includegraphics[width=1\textwidth]{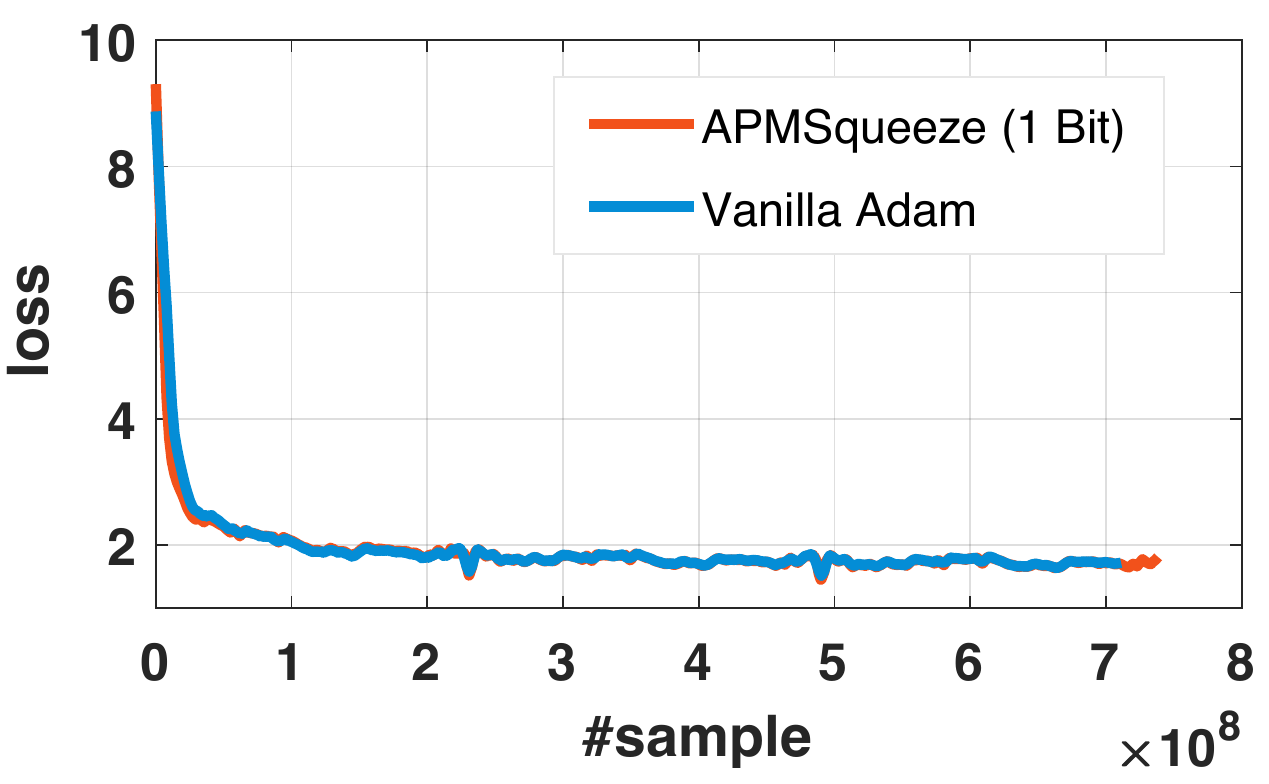}
\end{minipage}
}\quad
\subfigure[BERT-Large 128]{
\begin{minipage}[t]{0.4\linewidth}
\centering
\includegraphics[width=1\textwidth]{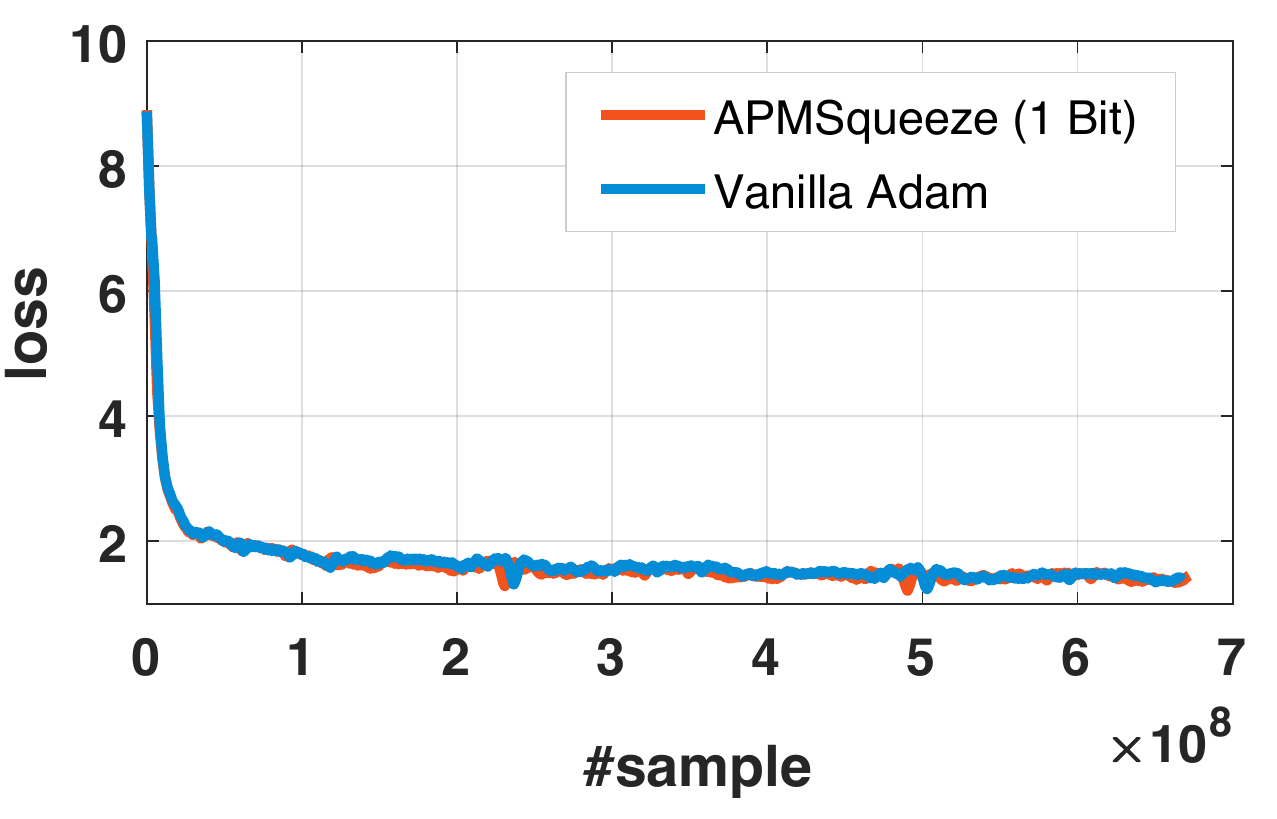}
\end{minipage}%
}
\centering
\vspace{-0.3cm}
\caption{Epoch-wise Convergence Speed (pretrain) for BERT using Sequence Length 128}\label{fig:bert}\vspace{-0.3cm}
\end{figure}

\subsection{ResNet on CIFAR10}\label{resnet}

\begin{figure}[t]
\centering
\subfigure[Training Loss]{
\begin{minipage}[t]{0.35\linewidth}
\centering
\includegraphics[width=1\textwidth]{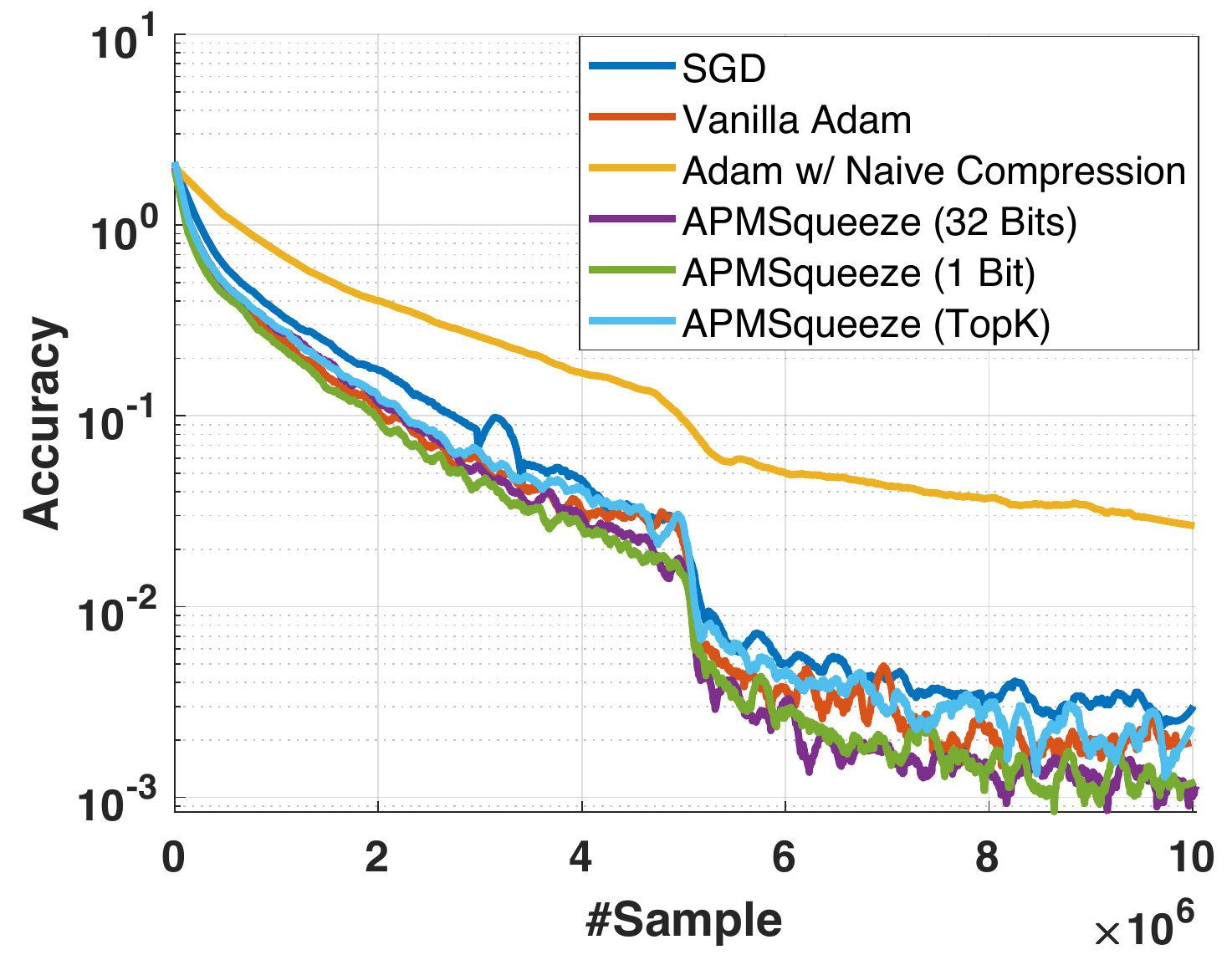}
\end{minipage}
}\quad
\subfigure[Testing Accuracy]{
\begin{minipage}[t]{0.35\linewidth}
\centering
\includegraphics[width=1\textwidth]{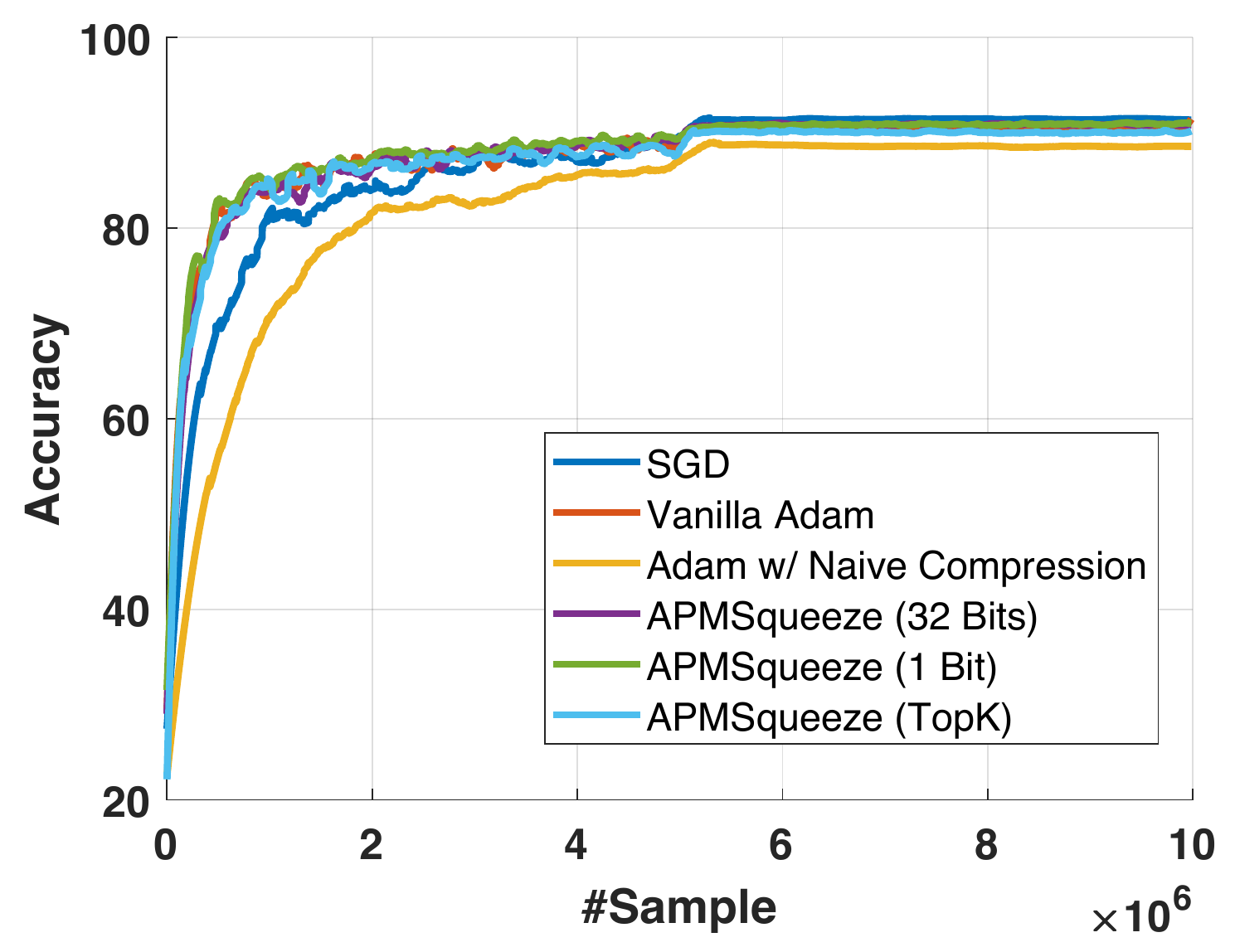}
\end{minipage}%
}
\centering
\caption{Epoch-wise Convergence Speed for ResNet-18 }\label{fig:resnet}
\vspace{-0.6cm}
\end{figure}

\paragraph{Dataset} We benchmark the algorithms using a standard image
classification task: training CIFAR10 using ResNet-18 \citep{7780459}. This dataset has a training
set of 50,000 images and a test set of 10,000 images, where each image is given
one of the 10 labels.

\paragraph{Hardware} We run the experiments on $8$ 1080Ti GPUs, each GPU is used as one worker. Here the batch-size on each worker is $128$, therefore the total batch-size is $1024$.

\vspace{-0.3cm}
\paragraph{Implementations and setups} We evaluate five implementations for comparison:
\begin{enumerate}
\item \textbf{{Original Adam}.} This is essentially the original Adam \citep{adam} implementation. We grid search the learning rate $\in\{1\times 10^{-2},1\times 10^{-3},1\times 10^{-4},1\times 10^{-5}\} $, and choose the best learning rate $1e-4 $.

\item \textbf{$\DS$ (Compressed).} This is essentially our $\DS$ algorithm. We  use $13$ epochs for Adam pre-conditioned and the total training takes $200$ epochs.  The learning rate  $1\times 10^{-4}$ is the same as the original Adam.

\item \textbf{$\DS$ (Uncompressed).} This is an uncompressed version of the $\DS $, which means we do not compress after the Adam pre-conditioned stage but still  stops updating $\v$.  We  use $13$ epochs for Adam pre-conditioned and the total training takes $200$ epochs.  The learning rate  $1\times 10^{-4}$ is the same as the original Adam.

\item \textbf{$\texttt{APGSqueeze}$.} Instead of communicating the momentum $\m_t$, we use the Error-Compensate compression strategy for communicating the gradient $\g_t$. We  use $13$ epochs for Adam pre-conditioned and the total training takes $200$ epochs.  The learning rate  $1\times 10^{-4}$ is the same as the original Adam.

\item \textbf{SGD.} This is vanilla SGD without compression. We grid search the learning rate $\in\{5\times 10^{-1},1\times 10^{-1},1\times 10^{-2},1\times 10^{-3}\} $, and choose the best learning rate $1\times 10^{-1} $.
\end{enumerate}
Notice that   the learning rate is decayed into $10\%$ of the original after every $100$ epochs.

\paragraph{Convergence Results}
In Figure~\ref{fig:resnet}, we report the sample-wise convergence result for each algorithm.
 We shall see that after the Adam pre-conditioned stage, the training efficiency remains almost the same for the compressed training and the uncompressed one, while the communication is reduced $97\%$ of the original Adam. 

\textbf{Conclusions}
In this paper, we propose an error compensated Adam preconditioned momentum SGD algorithm ($\OA$), which enables a compressed communication but could still achieve almost the same convergence rate with Adam. As a result, the communication overhead can be reduced into $3\%$ of the original and substantially accelerate the training speed under limited network bandwidth. Our theoretical analysis that $\OA$ admits a linear speed w.r.t the number of workers in the network, and is robust to any compression method. We validate the performance of $\OA$ empirically for training BERT-base, BERT-large and ResNet-18.

%
%
%

\section{Broader Impact}
In this paper, we propose a communication efficient algorithm that admits almost the same convergence rate with Adam, but achieves 10-30$\times$ speedup corresponding to the communication overhead. 

Adam, which uses a more complicated updating rule than SGD,  has shown to be a very efficient optimizer due to its fast convergence speed, and is even necessary for some large-scale machine learning tasks, such as BERT. With the increasing number of parameters in model machine learning models, it has become necessary to use large-scale (hundreds or even thousands of workers) parallel training, and the communication overhead could be comparable or even substantially overweight the computation cost for this large-scale training task. Therefore it is very important to design a communication efficient algorithm for Adam.
However, there was no communication efficient algorithm proposed before due to the  intrinsic non-linearity of Adam. To our knowledge, we are the first work that solve this problem. We want to emphasize that our method (using Adam as preconditioned warmup) could not only reduce a great amount ($97\%$) of communication cost, but also can be transferred to other applications of Adam, such as decentralized training, which is also a widely used method for reducing the communication overhead when the network latency is high.

\bibliographystyle{abbrvnat}
\bibliography{reference}

\begin{thebibliography}{36}
\providecommand{\natexlab}[1]{#1}
\providecommand{\url}[1]{\texttt{#1}}
\expandafter\ifx\csname urlstyle\endcsname\relax
  \providecommand{\doi}[1]{doi: #1}\else
  \providecommand{\doi}{doi: \begingroup \urlstyle{rm}\Url}\fi

\bibitem[Agarwal et~al.(2018)Agarwal, Suresh, Yu, Kumar, and
  McMahan]{Agarwal2018-hg}
N.~Agarwal, A.~T. Suresh, F.~X.~X. Yu, S.~Kumar, and B.~McMahan.
\newblock {cpSGD}: Communication-efficient and differentially-private
  distributed {SGD}.
\newblock In S.~Bengio, H.~Wallach, H.~Larochelle, K.~Grauman, N.~Cesa-Bianchi,
  and R.~Garnett, editors, \emph{Advances in Neural Information Processing
  Systems 31}, pages 7564--7575. Curran Associates, Inc., 2018.

\bibitem[Alistarh et~al.(2017)Alistarh, Grubic, Li, Tomioka, and
  Vojnovic]{Alistarh2017-yh}
D.~Alistarh, D.~Grubic, J.~Li, R.~Tomioka, and M.~Vojnovic.
\newblock {QSGD}: {Communication-Efficient} {SGD} via gradient quantization and
  encoding.
\newblock In I.~Guyon, U.~V. Luxburg, S.~Bengio, H.~Wallach, R.~Fergus,
  S.~Vishwanathan, and R.~Garnett, editors, \emph{Advances in Neural
  Information Processing Systems 30}, pages 1709--1720. Curran Associates,
  Inc., 2017.

\bibitem[Basu et~al.(2019)Basu, Data, Karakus, and Diggavi]{NIPS2019_9610}
D.~Basu, D.~Data, C.~Karakus, and S.~Diggavi.
\newblock Qsparse-local-sgd: Distributed sgd with quantization, sparsification
  and local computations.
\newblock In H.~Wallach, H.~Larochelle, A.~Beygelzimer, F.~d~Alch\'{e}-Buc,
  E.~Fox, and R.~Garnett, editors, \emph{Advances in Neural Information
  Processing Systems 32}, pages 14695--14706. Curran Associates, Inc., 2019.

\bibitem[Bernstein et~al.(2018)Bernstein, Zhao, Azizzadenesheli, and
  Anandkumar]{Bernstein:2018aa}
J.~Bernstein, J.~Zhao, K.~Azizzadenesheli, and A.~Anandkumar.
\newblock signsgd with majority vote is communication efficient and byzantine
  fault tolerant.
\newblock 10 2018.

\bibitem[Chaturapruek et~al.(2015)Chaturapruek, Duchi, and
  R\'{e}]{NIPS2015_6031}
S.~Chaturapruek, J.~C. Duchi, and C.~R\'{e}.
\newblock Asynchronous stochastic convex optimization: the noise is in the
  noise and sgd don t care.
\newblock In C.~Cortes, N.~D. Lawrence, D.~D. Lee, M.~Sugiyama, and R.~Garnett,
  editors, \emph{Advances in Neural Information Processing Systems 28}, pages
  1531--1539. Curran Associates, Inc., 2015.

\bibitem[Devlin et~al.(2019)Devlin, Chang, Lee, and Toutanova]{bert}
J.~Devlin, M.-W. Chang, K.~Lee, and K.~Toutanova.
\newblock Bert: Pre-training of deep bidirectional transformers for language
  understanding.
\newblock In \emph{NAACL-HLT}, 2019.

\bibitem[Duchi et~al.(2011)Duchi, Hazan, and Singer]{JMLR:v12:duchi11a}
J.~Duchi, E.~Hazan, and Y.~Singer.
\newblock Adaptive subgradient methods for online learning and stochastic
  optimization.
\newblock \emph{Journal of Machine Learning Research}, 12\penalty0
  (61):\penalty0 2121--2159, 2011.
\newblock URL \url{http://jmlr.org/papers/v12/duchi11a.html}.

\bibitem[{He} et~al.(2016){He}, {Zhang}, {Ren}, and {Sun}]{7780459}
K.~{He}, X.~{Zhang}, S.~{Ren}, and J.~{Sun}.
\newblock Deep residual learning for image recognition.
\newblock In \emph{2016 IEEE Conference on Computer Vision and Pattern
  Recognition (CVPR)}, pages 770--778, 2016.

\bibitem[Ivkin et~al.(2019)Ivkin, Rothchild, Ullah, braverman, Stoica, and
  Arora]{NIPS2019_9473}
N.~Ivkin, D.~Rothchild, E.~Ullah, V.~braverman, I.~Stoica, and R.~Arora.
\newblock Communication-efficient distributed sgd with sketching.
\newblock In H.~Wallach, H.~Larochelle, A.~Beygelzimer, F.~d~Alch\'{e}-Buc,
  E.~Fox, and R.~Garnett, editors, \emph{Advances in Neural Information
  Processing Systems 32}, pages 13144--13154. Curran Associates, Inc., 2019.

\bibitem[Jiang and Agrawal(2018)]{NIPS2018_7519}
P.~Jiang and G.~Agrawal.
\newblock A linear speedup analysis of distributed deep learning with sparse
  and quantized communication.
\newblock In S.~Bengio, H.~Wallach, H.~Larochelle, K.~Grauman, N.~Cesa-Bianchi,
  and R.~Garnett, editors, \emph{Advances in Neural Information Processing
  Systems 31}, pages 2530--2541. Curran Associates, Inc., 2018.

\bibitem[Kingma and Ba(2014)]{adam}
D.~Kingma and J.~Ba.
\newblock Adam: A method for stochastic optimization.
\newblock \emph{International Conference on Learning Representations}, 12 2014.

\bibitem[Kingma and Ba(2015)]{Kingma2015AdamAM}
D.~P. Kingma and J.~Ba.
\newblock Adam: A method for stochastic optimization.
\newblock \emph{CoRR}, abs/1412.6980, 2015.

\bibitem[Koloskova* et~al.(2020)Koloskova*, Lin*, Stich, and
  Jaggi]{Koloskova*2020Decentralized}
A.~Koloskova*, T.~Lin*, S.~U. Stich, and M.~Jaggi.
\newblock Decentralized deep learning with arbitrary communication compression.
\newblock In \emph{International Conference on Learning Representations}, 2020.
\newblock URL \url{https://openreview.net/forum?id=SkgGCkrKvH}.

\bibitem[Li et~al.(2018)Li, Yu, Li, Avestimehr, Kim, and
  Schwing]{NIPS2018_8028}
Y.~Li, M.~Yu, S.~Li, S.~Avestimehr, N.~S. Kim, and A.~Schwing.
\newblock Pipe-sgd: A decentralized pipelined sgd framework for distributed
  deep net training.
\newblock In S.~Bengio, H.~Wallach, H.~Larochelle, K.~Grauman, N.~Cesa-Bianchi,
  and R.~Garnett, editors, \emph{Advances in Neural Information Processing
  Systems 31}, pages 8056--8067. Curran Associates, Inc., 2018.

\bibitem[Lian et~al.(2017)Lian, Zhang, Zhang, Hsieh, Zhang, and
  Liu]{Lian2017-ni}
X.~Lian, C.~Zhang, H.~Zhang, C.-J. Hsieh, W.~Zhang, and J.~Liu.
\newblock Can decentralized algorithms outperform centralized algorithms? a
  case study for decentralized parallel stochastic gradient descent.
\newblock In I.~Guyon, U.~V. Luxburg, S.~Bengio, H.~Wallach, R.~Fergus,
  S.~Vishwanathan, and R.~Garnett, editors, \emph{Advances in Neural
  Information Processing Systems 30}, pages 5330--5340. Curran Associates,
  Inc., 2017.

\bibitem[Luo et~al.(2019)Luo, Xiong, and Liu]{luo2018adaptive}
L.~Luo, Y.~Xiong, and Y.~Liu.
\newblock Adaptive gradient methods with dynamic bound of learning rate.
\newblock In \emph{International Conference on Learning Representations}, 2019.
\newblock URL \url{https://openreview.net/forum?id=Bkg3g2R9FX}.

\bibitem[{Phuong} and {Phong}(2020)]{9051706}
T.~T. {Phuong} and L.~T. {Phong}.
\newblock Distributed sgd with flexible gradient compression.
\newblock \emph{IEEE Access}, 8:\penalty0 64707--64717, 2020.

\bibitem[Seide et~al.(2014)Seide, Fu, Droppo, Li, and Yu]{1-bitexp}
F.~Seide, H.~Fu, J.~Droppo, G.~Li, and D.~Yu.
\newblock 1-bit stochastic gradient descent and application to data-parallel
  distributed training of speech dnns.
\newblock In \emph{Interspeech 2014}, September 2014.

\bibitem[Shen et~al.(2018)Shen, Mokhtari, Zhou, Zhao, and
  Qian]{pmlr-v80-shen18a}
Z.~Shen, A.~Mokhtari, T.~Zhou, P.~Zhao, and H.~Qian.
\newblock Towards more efficient stochastic decentralized learning: Faster
  convergence and sparse communication.
\newblock In J.~Dy and A.~Krause, editors, \emph{Proceedings of the 35th
  International Conference on Machine Learning}, volume~80 of \emph{Proceedings
  of Machine Learning Research}, pages 4624--4633, Stockholmsm{\"a}ssan,
  Stockholm Sweden, 10--15 Jul 2018. PMLR.

\bibitem[{Shi} et~al.(2019){Shi}, {Wang}, {Zhao}, {Tang}, {Wang}, {Huang}, and
  {Chu}]{8884924}
S.~{Shi}, Q.~{Wang}, K.~{Zhao}, Z.~{Tang}, Y.~{Wang}, X.~{Huang}, and X.~{Chu}.
\newblock A distributed synchronous sgd algorithm with global top-k
  sparsification for low bandwidth networks.
\newblock In \emph{2019 IEEE 39th International Conference on Distributed
  Computing Systems (ICDCS)}, pages 2238--2247, 2019.

\bibitem[Spring et~al.(2019)Spring, Kyrillidis, Mohan, and
  Shrivastava]{Spring2019-ep}
R.~Spring, A.~Kyrillidis, V.~Mohan, and A.~Shrivastava.
\newblock Compressing gradient optimizers via {Count-Sketches}.
\newblock \emph{Proceedings of the 36th International Conference on Machine
  Learning}, 97:\penalty0 5946--5955, 2019.

\bibitem[Stich et~al.(2018)Stich, Cordonnier, and Jaggi]{martinmemory}
S.~U. Stich, J.-B. Cordonnier, and M.~Jaggi.
\newblock Sparsified sgd with memory.
\newblock In S.~Bengio, H.~Wallach, H.~Larochelle, K.~Grauman, N.~Cesa-Bianchi,
  and R.~Garnett, editors, \emph{Advances in Neural Information Processing
  Systems 31}, pages 4447--4458. Curran Associates, Inc., 2018.

\bibitem[Sun et~al.(2019)Sun, Chen, Giannakis, and Yang]{NIPS2019_8598}
J.~Sun, T.~Chen, G.~Giannakis, and Z.~Yang.
\newblock Communication-efficient distributed learning via lazily aggregated
  quantized gradients.
\newblock In H.~Wallach, H.~Larochelle, A.~Beygelzimer, F.~d~Alch\'{e}-Buc,
  E.~Fox, and R.~Garnett, editors, \emph{Advances in Neural Information
  Processing Systems 32}, pages 3370--3380. Curran Associates, Inc., 2019.

\bibitem[Tang et~al.(2019)Tang, Yu, Lian, Zhang, and Liu]{tangdouble}
H.~Tang, C.~Yu, X.~Lian, T.~Zhang, and J.~Liu.
\newblock $\texttt{DoubleSqueeze}$: Parallel stochastic gradient descent with
  double-pass error-compensated compression.
\newblock In K.~Chaudhuri and R.~Salakhutdinov, editors, \emph{Proceedings of
  the 36th International Conference on Machine Learning}, volume~97 of
  \emph{Proceedings of Machine Learning Research}, pages 6155--6165, Long
  Beach, California, USA, 09--15 Jun 2019. PMLR.

\bibitem[Tieleman and Hinton(2011)]{rmsprop}
T.~Tieleman and G.~Hinton.
\newblock Rmsprop: Divide the gradient by a running average of its recent
  magnitude.
\newblock \emph{COURSERA: Neural networks for machine learning}, 2011.

\bibitem[Vogels et~al.(2019)Vogels, Karimireddy, and Jaggi]{NIPS2019_9571}
T.~Vogels, S.~P. Karimireddy, and M.~Jaggi.
\newblock Powersgd: Practical low-rank gradient compression for distributed
  optimization.
\newblock In H.~Wallach, H.~Larochelle, A.~Beygelzimer, F.~d~Alch\'{e}-Buc,
  E.~Fox, and R.~Garnett, editors, \emph{Advances in Neural Information
  Processing Systems 32}, pages 14259--14268. Curran Associates, Inc., 2019.

\bibitem[Wang et~al.(2018)Wang, Singh, Michael, Hill, Levy, and Bowman]{glue}
A.~Wang, A.~Singh, J.~Michael, F.~Hill, O.~Levy, and S.~Bowman.
\newblock {GLUE}: A multi-task benchmark and analysis platform for natural
  language understanding.
\newblock In \emph{Proceedings of the 2018 {EMNLP} Workshop {B}lackbox{NLP}:
  Analyzing and Interpreting Neural Networks for {NLP}}, pages 353--355,
  Brussels, Belgium, Nov. 2018. Association for Computational Linguistics.
\newblock \doi{10.18653/v1/W18-5446}.
\newblock URL \url{https://www.aclweb.org/anthology/W18-5446}.

\bibitem[Wangni et~al.(2018)Wangni, Wang, Liu, and Zhang]{Wangni2018-ux}
J.~Wangni, J.~Wang, J.~Liu, and T.~Zhang.
\newblock Gradient sparsification for {Communication-Efficient} distributed
  optimization.
\newblock In S.~Bengio, H.~Wallach, H.~Larochelle, K.~Grauman, N.~Cesa-Bianchi,
  and R.~Garnett, editors, \emph{Advances in Neural Information Processing
  Systems 31}, pages 1299--1309. Curran Associates, Inc., 2018.

\bibitem[Wen et~al.(2017)Wen, Xu, Yan, Wu, Wang, Chen, and Li]{NIPS2017_6749}
W.~Wen, C.~Xu, F.~Yan, C.~Wu, Y.~Wang, Y.~Chen, and H.~Li.
\newblock Terngrad: Ternary gradients to reduce communication in distributed
  deep learning.
\newblock In I.~Guyon, U.~V. Luxburg, S.~Bengio, H.~Wallach, R.~Fergus,
  S.~Vishwanathan, and R.~Garnett, editors, \emph{Advances in Neural
  Information Processing Systems 30}, pages 1509--1519. Curran Associates,
  Inc., 2017.

\bibitem[Ye and Abbe(2018)]{Ye2018-mf}
M.~Ye and E.~Abbe.
\newblock {Communication-Computation} efficient gradient coding.
\newblock \emph{Proceedings of the 35th International Conference on Machine
  Learning}, 80:\penalty0 5610--5619, 2018.

\bibitem[Yu et~al.(2019{\natexlab{a}})Yu, Tang, Renggli, Kassing, Singla,
  Alistarh, Zhang, and Liu]{pmlr-v97-yu19f}
C.~Yu, H.~Tang, C.~Renggli, S.~Kassing, A.~Singla, D.~Alistarh, C.~Zhang, and
  J.~Liu.
\newblock Distributed learning over unreliable networks.
\newblock In K.~Chaudhuri and R.~Salakhutdinov, editors, \emph{Proceedings of
  the 36th International Conference on Machine Learning}, volume~97 of
  \emph{Proceedings of Machine Learning Research}, pages 7202--7212, Long
  Beach, California, USA, 09--15 Jun 2019{\natexlab{a}}. PMLR.
\newblock URL \url{http://proceedings.mlr.press/v97/yu19f.html}.

\bibitem[Yu et~al.(2019{\natexlab{b}})Yu, Wu, and Huang]{NIPS2019_8694}
Y.~Yu, J.~Wu, and L.~Huang.
\newblock Double quantization for communication-efficient distributed
  optimization.
\newblock In H.~Wallach, H.~Larochelle, A.~Beygelzimer, F.~d~Alch\'{e}-Buc,
  E.~Fox, and R.~Garnett, editors, \emph{Advances in Neural Information
  Processing Systems 32}, pages 4438--4449. Curran Associates, Inc.,
  2019{\natexlab{b}}.

\bibitem[Zeiler(2012)]{DBLP:journals/corr/abs-1212-5701}
M.~D. Zeiler.
\newblock {ADADELTA:} an adaptive learning rate method.
\newblock \emph{CoRR}, abs/1212.5701, 2012.
\newblock URL \url{http://arxiv.org/abs/1212.5701}.

\bibitem[Zhang et~al.(2017)Zhang, Li, Kara, Alistarh, Liu, and
  Zhang]{pmlr-v70-zhang17e}
H.~Zhang, J.~Li, K.~Kara, D.~Alistarh, J.~Liu, and C.~Zhang.
\newblock {Z}ip{ML}: Training linear models with end-to-end low precision, and
  a little bit of deep learning.
\newblock In D.~Precup and Y.~W. Teh, editors, \emph{Proceedings of the 34th
  International Conference on Machine Learning}, volume~70 of \emph{Proceedings
  of Machine Learning Research}, pages 4035--4043, International Convention
  Centre, Sydney, Australia, 06--11 Aug 2017. PMLR.

\bibitem[Zheng et~al.(2016)Zheng, Meng, Wang, Chen, Yu, Ma, and
  Liu]{DBLP:journals/corr/ZhengMWCYML16}
S.~Zheng, Q.~Meng, T.~Wang, W.~Chen, N.~Yu, Z.~Ma, and T.~Liu.
\newblock Asynchronous stochastic gradient descent with delay compensation for
  distributed deep learning.
\newblock \emph{CoRR}, abs/1609.08326, 2016.

\bibitem[Zheng et~al.(2019)Zheng, Huang, and Kwok]{NIPS2019_9321}
S.~Zheng, Z.~Huang, and J.~Kwok.
\newblock Communication-efficient distributed blockwise momentum sgd with
  error-feedback.
\newblock In H.~Wallach, H.~Larochelle, A.~Beygelzimer, F.~d~Alch\'{e}-Buc,
  E.~Fox, and R.~Garnett, editors, \emph{Advances in Neural Information
  Processing Systems 32}, pages 11450--11460. Curran Associates, Inc., 2019.

\end{thebibliography}

\newpage
\onecolumn
{\Huge Supplementary}
\section{Proof to the Updating Form}
Since our algorithm is equivalent to running a parameter-server prototype communication on each chunk of the gradient, so below we will assume a parameter-server model (which means the tensor is not required to be divided into $n$ chunks) for simplicity.

According to the algorithm description in Section~\ref{alg:description}, at iteration $t+1$, the updating step of the momentum term $\bm{m}_{t+1}$ can be divided into two steps:
\begin{enumerate}
\item Local Update and Compress: each worker locally update $\bm{m}_t$ and use the error-compensate strategy for compressing.
\begin{align*}
\m_{t}^{(i)} =& \beta \bm{m}_t +  (1-\beta)\bm{g}_t^{(i)}\\
\m_{t+\frac{1}{2}}^{(i)} = & C_{\omega}[\m_{t}^{(i)} + \bm{\delta}_t^{(i)}]\\
\bm{\delta}_{t+1}^{(i)} = & \m_{t}^{(i)} + \bm{\delta}_t^{(i)} - \m_{t+\frac{1}{2}}^{(i)}.
\end{align*}
\item All workers send its  $\m_{t+\frac{1}{2}}^{(i)}$ to the server. The server takes the average over them and compress it again using error-compensation.
\begin{align*}
\m_{t+\frac{1}{2}} =& \frac{1}{n}\sum_{i=1}^n \m_{t+\frac{1}{2}}^{(i)}\\
\m_{t+1} = & C_{\omega}[\m_{t+\frac{1}{2}} + \bm{\delta}_t]\\
\bm{\delta}_{t+1} = & \m_{t+\frac{1}{2}} + \bm{\delta}_t - \m_{t+1}.
\end{align*}
\item The server broadcast $\m_{t+1}$ to all workers, and all workers update the local model according to
\begin{align*}
\bm{x}_{t+1} = \bm{x}_{t} - \bm{\gamma}\m_{t+1}\oslash \sqrt{\bm{v}_{T_{w}}^2}.
\end{align*}

\end{enumerate}
So actually the updating rule above can be summarized as
\begin{align*}
\m_{t+1} =& C_{\omega}[\m_{t+\frac{1}{2}} + \bm{\delta}_t]\\
= & \m_{t+\frac{1}{2}} + \bm{\delta}_t - \bm{\delta}_{t+1}\quad\text{(from the definition of $\bm{\delta}_{t+1}$)}\\
= & \frac{1}{n}\sum_{i=1}^n C_{\omega}[\m_{t}^{(i)} + \bm{\delta}_t^{(i)}]+ \bm{\delta}_t - \bm{\delta}_{t+1}\\
= & \frac{1}{n}\sum_{i=1}^n\left( \m_{t}^{(i)} + \bm{\delta}_t^{(i)} - \bm{\delta}_{t+1}^{(i)} \right) + \bm{\delta}_t - \bm{\delta}_{t+1}\quad\text{(from the definition of $\bm{\delta}_{t+1}^{(i)}$)}\\
= & \beta \m_t + \frac{1-\beta}{n}\sum_{i=1}^n \bm{g}_t^{(i)} + \left( \frac{1}{n}\sum_{i=1}^n \bm{\delta}_t^{(i)} + \bm{\delta}_t\right) - \left( \frac{1}{n}\sum_{i=1}^n \bm{\delta}_{t+1}^{(i)} + \bm{\delta}_{t+1}\right).
\end{align*}
Denote
\begin{align*}
\overline{\bm{g}}_t = & \frac{1}{n}\sum_{i=1}^n \bm{g}_t^{(i)}\\
\overline{\bm{\delta}}_{t} = & \frac{1}{n}\sum_{i=1}^n \bm{\delta}_t^{(i)} + \bm{\delta}_t,
\end{align*}
the update rule of $\m_{t}$ can be summarized as
\begin{align*}
\m_t = \beta\m_{t-1} + (1-\beta)\overline{\bm{g}}_t + \overline{\bm{\delta}}_{t-1} - \overline{\bm{\delta}}_{t},
\end{align*}
and
\begin{align*}
\x_{t+1} = \x_t - \gamma V\m_t,
\end{align*}
where $V= \text{diag}(1/\sqrt{v_1},1/\sqrt{v_2},\cdots,1/\sqrt{v_d})$ is the a diagonal matrix that spanned with $\v_{T_{w}}$.

\section{Proof to Theorem~\ref{theo:global}}
Notice that in for $\OA$, the learning rate for each coordinate is different. In order to simplify our analysis, we instead consider another function that is defined as
\begin{align*}
H(\z) = F(V^{\frac{1}{2}}\z),
\end{align*}
also
\begin{align*}
h(\z) = f(V^{\frac{1}{2}}\z),
\end{align*}
where ${V} $ is a diagonal matrix.

In this case we have
\begin{align*}
{V}^{\frac{1}{2}}\nabla f({V}^{\frac{1}{2}}\z) = \nabla h(\z).
\end{align*}
Therefore the updating rule of $\OA$ in the view of $h(\cdot)$ is
\begin{align*}
{V}^{\frac{1}{2}}\z_{t+1} = {V}^{\frac{1}{2}}\z_t - \gamma {V}^{\frac{1}{2}}\left({V}^{\frac{1}{2}}\m_t\right).
\end{align*}
It can be easily verified that
\begin{align*}
\bm{m}_t =& (1-\beta) \sum_{s=0}^t \beta^{t-s} \overline{\bm{g}}_s + \sum_{s=0}^t \beta^{t-s}( \overline{\bm{\delta}}_{s-1} - \overline{\bm{\delta}}_{s})\\
= & (1-\beta) \sum_{s=0}^t \beta^{t-s} \frac{1}{n}\sum_{i=1}^n\nabla F(V^{\frac{1}{2}}\z_t;\xi_t^{(i)}) + \sum_{s=0}^t \beta^{t-s}( \overline{\bm{\delta}}_{s-1} - \overline{\bm{\delta}}_{s})
\end{align*}
which means
\begin{align*}
{V}^{\frac{1}{2}}\bm{m}_t =& (1-\beta) \sum_{s=0}^t \beta^{t-s} \frac{1}{n}\sum_{i=1}^n{V}^{\frac{1}{2}}\nabla F(V^{\frac{1}{2}}\z_t;\xi_t^{(i)}) + \sum_{s=0}^t \beta^{t-s}{V}^{\frac{1}{2}}( \overline{\bm{\delta}}_{s-1} - \overline{\bm{\delta}}_{s})\\
= & (1-\beta) \sum_{s=0}^t \beta^{t-s} \frac{1}{n}\sum_{i=1}^n\nabla H(V^{\frac{1}{2}}\z_t;\xi_t^{(i)}) + \sum_{s=0}^t \beta^{t-s}{V}^{\frac{1}{2}}( \overline{\bm{\delta}}_{s-1} - \overline{\bm{\delta}}_{s})\\
= & (1-\beta) \sum_{s=0}^t \beta^{t-s} \overline{\bm{g}}_s(\z) + \sum_{s=0}^t \beta^{t-s}{V}^{\frac{1}{2}}( \overline{\bm{\delta}}_{s-1} - \overline{\bm{\delta}}_{s}),
\end{align*}
where $\overline{\bm{g}}_s(\z)$ is the corresponding averaged stochastic gradient computed in the view of loss function $h(\cdot)$.

Then, if we define $\m_t(\z) = {V}^{\frac{1}{2}}\m_t$, the updating rule of $\m_t(z)$ admits
\begin{align*}
\m_t(\z) = \beta\m_{t-1}(\z) + (1-\beta)\overline{\bm{g}}_t(\z) + {V}^{\frac{1}{2}}\overline{\bm{\delta}}_{t-1} - {V}^{\frac{1}{2}}\overline{\bm{\delta}}_{t}, \numberthis\label{supp:trans_eq1}
\end{align*}
and
\begin{align*}
{V}^{\frac{1}{2}}\z_{t+1} =& {V}^{\frac{1}{2}}\z_t - \gamma {V}^{\frac{1}{2}}\m_t(\z)\\
\z_{t+1} =& \z_t - \gamma \m_t(\z).\numberthis\label{supp:trans_eq2}
\end{align*}
From \eqref{supp:trans_eq1} and \eqref{supp:trans_eq2} we shall see that using different learning rate for each coordinate is equivalent to optimizing a new loss function defined on  scaling the original coordinate and using a uniform learning for all coordinates. Therefore below we first study the behavior of the error-compensated momentum SGD using a constant learning rate.

Below are some critical lemmas for the proof of Theorem~\ref{theo:global}.
\begin{lemma}\label{lemma:seq}

Given two non-negative sequences $\{a_t\}_{t=1}^{\infty}$ and $\{b_t\}_{t=1}^{\infty}$ that satisfying
\begin{equation}
a_t =  \sum_{s=1}^t\rho^{t-s}b_{s}, \numberthis \label{eqn1}
\end{equation}
with $\rho\in[0,1)$, we have
\begin{align*}
D_k:=\sum_{t=1}^{k}a_t^2 \leq & 
\frac{1}{(1-\rho)^2} \sum_{s=1}^kb_s^2.
\end{align*}
\end{lemma}

\begin{proof}
From the definition, we have
\begin{align*}
S_k= & \sum_{t=1}^{k}\sum_{s=1}^t\rho^{t-s}b_{s}
=  \sum_{s=1}^{k}\sum_{t=s}^k\rho^{t-s}b_{s}
=  \sum_{s=1}^{k}\sum_{t=0}^{k-s}\rho^{t}b_{s}
\leq  \sum_{s=1}^{k}{b_{s}\over 1-\rho}, \numberthis \label{eqn3}\\
D_k=  & \sum_{t=1}^{k}\sum_{s=1}^t\rho^{t-s}b_{s}\sum_{r=1}^t\rho^{t-r}b_{r}\\
= & \sum_{t=1}^{k}\sum_{s=1}^t\sum_{r=1}^t\rho^{2t-s-r}b_{s}b_{r} \\
\leq &  \sum_{t=1}^{k}\sum_{s=1}^t\sum_{r=1}^t\rho^{2t-s-r}{b_{s}^2+b_{r}^2\over2}\\
= & \sum_{t=1}^{k}\sum_{s=1}^t\sum_{r=1}^t\rho^{2t-s-r}b_{s}^2 \\
\leq  & {1\over 1-\rho}\sum_{t=1}^{k}\sum_{s=1}^t\rho^{t-s}b_{s}^2\\
\leq & {1\over (1-\rho)^2}\sum_{s=1}^{k}b_{s}^2, \quad \text{(due to \eqref{eqn3})}
\end{align*}
which completes the proof.
\end{proof}

\begin{lemma}\label{lemma:supp_main}
Under Assumption~\ref{ass:global}, for any sequence that follows the updating rule of
\begin{align*}
\x_{t+1} =& \x_t - \gamma \m_t\\
\m_t =& \beta\m_{t-1} + (1-\beta)\overline{\bm{g}}_t + \overline{\bm{\delta}}_{t-1} - \overline{\bm{\delta}}_{t},
\end{align*}
if 
\begin{align*}
\mathbb E \overline{\g}_t =  \nabla & f(\x_t),\quad\mathbb E \|\overline{\g}_t - \nabla f(\x_t)\|^2\leq \frac{\sigma^2}{n},\quad \mathbb E\|\overline{\bm{\delta}}_t\|^2\leq \epsilon^2,\quad \forall t,\\
&\|\nabla f(\bm{x}) - \nabla f(\bm{y}) \| \leq L \|\bm{x} - \bm{y} \|,\quad \forall \bm{x},\forall \bm{y},
\end{align*}
then we can guarantee that
\begin{align*}
&\left(1-\gamma L - \frac{2\gamma^2 L^2}{(1-\beta)^2} \right)\sum_{t=0}^T \mathbb E\|\nabla f(\bm{x}_t)\|^2\\
 \leq & \frac{2\mathbb E f(\bm{x}_{1}) - 2\mathbb Ef(\bm{x}^*)}{\gamma}   + \frac{6\gamma^2L^2\epsilon^2 T}{(1-\beta)^2}  +  \frac{L\gamma\sigma^2T}{n} + \frac{2\gamma^2L^2\sigma^2 T}{n(1-\beta)^2}
\end{align*}

\end{lemma}

\begin{proof}
Instead of investigating $\bm{x}_t$ directly, we introduce the following sequence
\begin{align*}
\bm{y}_t = \bm{x}_t - \frac{\gamma}{1-\beta}(\bm{m}_t + \overline{\bm{\delta}}_{t-1}) .
\end{align*}
The updating rule of $\bm{y}_t$ admits
\begin{align*}
\bm{y}_{t+1} - \bm{y}_t = & \bm{x}_{t+1} - \bm{x}_t - \frac{\gamma}{1- \beta}(\bm{m}_{t+1} - \bm{m}_t - \overline{\bm{\delta}}_{t+1} + \overline{\bm{\delta}}_t) \\
= & -\gamma \bm{m}_t  - \frac{\gamma}{1-\beta}(\beta \bm{m_t} + (1-\beta)\bm{g}_t + \overline{\bm{\delta}}_{t-1} - \overline{\bm{\delta}}_{t} - \bm{m}_t + \overline{\bm{\delta}}_{t} - \overline{\bm{\delta}}_{t-1})\\
= & -\gamma\bm{g}_t.
\end{align*}
Since $f(\cdot)$ is with L-Lipschitzian, we have
\begin{align*}
\mathbb E f(\bm{y}_{t+1}) - \mathbb Ef(\bm{y}_{t}) \leq & \mathbb E\left\langle \nabla f(\bm{y}_t), \bm{y}_{t+1} - \bm{y}_t \right\rangle + \frac{L}{2}\mathbb E\left\|\bm{y}_{t+1} - \bm{y}_t \right\|^2\\
=& -\gamma \mathbb E\left\langle \nabla f(\bm{y}_t), \bm{g}_t\right\rangle + \frac{L\gamma^2}{2}\mathbb E\|\bm{g}_t\|^2\\
=& -\gamma \mathbb E\left\langle \nabla f(\bm{y}_t), \nabla f(\bm{x}_t)\right\rangle + \frac{L\gamma^2}{2}\mathbb E\|\bm{g}_t\|^2\\
= & -\frac{\gamma}{2} \mathbb E\|\nabla f(\bm{x}_t)\|^2 - \frac{\gamma}{2}\mathbb E\|\nabla f(\bm{y}_t)\|^2 + \frac{\gamma}{2} \mathbb E\|\nabla f(\bm{x}_t) - \nabla f(\bm{y}_t)\|^2+ \frac{L\gamma^2}{2}\mathbb E\|\bm{g}_t\|^2\\
\leq & -\frac{\gamma}{2} \mathbb E\|\nabla f(\bm{x}_t)\|^2 + \frac{\gamma L^2}{2}\mathbb E\|\bm{x}_t - \bm{y}_t \|^2+ \frac{L\gamma^2}{2}\mathbb E\|\bm{g}_t\|^2\\
= & -\frac{\gamma}{2} \mathbb E\|\nabla f(\bm{x}_t)\|^2 + \frac{\gamma^3 L^2}{2}\mathbb E\left\|\frac{\bm{m}_t}{1-\beta} + \frac{\overline{\bm{\delta}}_{t-1}}{1-\beta} \right\|^2+ \frac{L\gamma^2}{2}\mathbb E\|\bm{g}_t\|^2\\
\leq & -\frac{\gamma}{2} \mathbb E\|\nabla f(\bm{x}_t)\|^2 + \frac{\gamma^3 L^2}{(1-\beta)^2}\mathbb E\|\bm{m}_t\|^2 + \frac{\gamma^3L^2}{(1-\beta)^2}\mathbb E\|\overline{\bm{\delta}}_{t-1}\|^2  +  \frac{L\gamma^2}{2}\mathbb E\|\bm{g}_t\|^2\\
\leq & -\frac{\gamma}{2} \mathbb E\|\nabla f(\bm{x}_t)\|^2 + \frac{\gamma^3 L^2}{(1-\beta)^2}\mathbb E\|\bm{m}_t\|^2 + \frac{\gamma^3L^2\epsilon^2}{(1-\beta)^2}  +  \frac{L\gamma^2}{2}\mathbb E\|\bm{g}_t\|^2\\
\leq & -\frac{\gamma}{2} \mathbb E\|\nabla f(\bm{x}_t)\|^2 + \frac{\gamma^3 L^2}{(1-\beta)^2}\mathbb E\|\bm{m}_t\|^2 + \frac{\gamma^3L^2\epsilon^2}{(1-\beta)^2}  +  \frac{L\gamma^2}{2}\mathbb E\|\nabla f(\bm{x}_t)\|^2 + \frac{L\gamma^2\sigma^2}{2n}.
\end{align*}
Summing up the equation above from $t=0$ to $t=T$ we get
\begin{align*}
\mathbb E f(\bm{y}_{T+1}) - \mathbb Ef(\bm{y}_{0}) \leq -\frac{(1-\gamma L)\gamma}{2}\sum_{t=0}^T \mathbb E\|\nabla f(\bm{x}_t)\|^2 + \frac{\gamma^3 L^2}{(1-\beta)^2}\sum_{t=0}^T\mathbb E\|\bm{m}_t\|^2 + \frac{\gamma^3L^2\epsilon^2 T}{(1-\beta)^2}  +  \frac{L\gamma^2\sigma^2 T}{2n},
\end{align*}
which can be rewritten into
\begin{align*}
(1-\gamma L)\sum_{t=0}^T \mathbb E\|\nabla f(\bm{x}_t)\|^2 \leq & \frac{2\mathbb E f(\bm{y}_{0}) - 2\mathbb Ef(\bm{y}_{T+1})}{\gamma}  +  \frac{2\gamma^2 L^2}{(1-\beta)^2}\sum_{t=0}^T\mathbb E\|\bm{m}_t\|^2 + \frac{2\gamma^2L^2\epsilon^2 T}{(1-\beta)^2}  +  \frac{L\gamma\sigma^2 T}{n}.\numberthis\label{supp:final_eq1}
\end{align*}

Notice that we have
\begin{align*}
\bm{m}_t =& (1-\beta) \sum_{s=0}^t \beta^{t-s} \overline{\bm{g}}_s + \sum_{s=0}^t \beta^{t-s}( \overline{\bm{\delta}}_{s-1} - \overline{\bm{\delta}}_{s})
\end{align*}
which by using Lemma~\ref{lemma:seq}, we have
\begin{align*}
\sum_{t=0}^T\|\bm{m}_t\|^2 \leq \sum_{t=0}^T \|\bm{g}_t\|^2 + \frac{2}{(1-\beta)^2}\sum_{t=0}^T \|\overline{\bm{\delta}}_t\|^2 \leq \sum_{t=0}^T \|\nabla f(\bm{x}_t)\|^2 + \frac{\sigma^2T}{n} + \frac{2\epsilon^2T}{(1-\beta)^2} . \numberthis\label{supp:final_nound_mt}
\end{align*}
Combing \eqref{supp:final_eq1} and \eqref{supp:final_nound_mt} together we get
\begin{align*}
&\left(1-\gamma L - \frac{2\gamma^2 L^2}{(1-\beta)^2} \right)\sum_{t=0}^T \mathbb E\|\nabla f(\bm{x}_t)\|^2\\
 \leq & \frac{2\mathbb E f(\bm{y}_{0}) - 2\mathbb Ef(\bm{y}_{T+1})}{\gamma}   + \frac{6\gamma^2L^2\epsilon^2 T}{(1-\beta)^2}  +  \frac{L\gamma\sigma^2T}{n} + \frac{2\gamma^2L^2\sigma^2 T}{n(1-\beta)^2}\\
 \leq & \frac{2\mathbb E f(\bm{x}_{1}) - 2\mathbb Ef(\bm{x}^*)}{\gamma}   + \frac{6\gamma^2L^2\epsilon^2 T}{(1-\beta)^2}  +  \frac{L\gamma\sigma^2T}{n} + \frac{2\gamma^2L^2\sigma^2 T}{n(1-\beta)^2}.
\end{align*}
\end{proof}

\paragraph{Proof to Theorem~\ref{theo:global}} Since using a per-coordinate learning rate for loss function $f(\cdot)$ is equivalent to use a constant learning for all coordinates but for loss function $h(\cdot)$, the only two thing that change are
\begin{itemize}
\item \textbf{Different L-Lipschitzian coefficient}: the L-Lipschitzian coefficient for $h(\cdot)$ is
\begin{align*}
\|\nabla h(\x) - \nabla h(\bm{y})\|^2 =& \left\|V^{\frac{1}{2}} \nabla f(V^{\frac{1}{2}}\x) - V^{\frac{1}{2}} \nabla f(V^{\frac{1}{2}}\bm{y}) \right\|^2 \\
= & \left\| \nabla f(V^{\frac{1}{2}}\x) - \nabla f(V^{\frac{1}{2}}\bm{y}) \right\|^2_V\\
\leq & L^2\left\| V^{\frac{1}{2}}\x - V^{\frac{1}{2}}\bm{y} \right\|^2_V\\
=& L^2\|\x - \bm{y}\|^2_{V^2}\\
\leq & L^2V_{\max}^2\|\x - \bm{y}\|^2.
\end{align*}
Therefore the effective L-Lipschitzian coefficient of $h(\x)$ is $LV_{\max}$
\item \textbf{Different definition of $\overline{\bm{\delta}}_t$}: from \eqref{supp:trans_eq1} we shall see that actually the compression error in the view of $h(\cdot)$ is $V^{\frac{1}{2}}\overline{\bm{\delta}}_t $, so in this case we have
\begin{align*}
\mathbb E\|V^{\frac{1}{2}}\overline{\bm{\delta}}_t\|^2 \leq V_{\max}\epsilon^2
\end{align*}

\end{itemize}
\begin{proof}
From Lemma~\ref{lemma:supp_main}, we have
\begin{align*}
&\left(1-\gamma L - \frac{2\gamma^2 L^2V_{\max}^2}{(1-\beta)^2} \right)\sum_{t=0}^T \mathbb E\|\nabla h(\bm{z}_t)\|^2\\
 \leq & \frac{2\mathbb E f(\bm{x}_{0}) - 2\mathbb Ef(\bm{x}^*)}{\gamma}   + \frac{6\gamma^2L^2\epsilon^2V_{\max}^3 T}{(1-\beta)^2}  +  \frac{L\gamma V_{\max}\sigma^2T}{n} + \frac{2\gamma^2L^2\sigma^2 V_{\max}^2 T}{n(1-\beta)^2},
\end{align*}
which by using $\nabla h(\bm{z}_t) = V^{\frac{1}{2}}\nabla f(\x_t) $, it becomes
\begin{align*}
&\left(1-\gamma LV_{\max} - \frac{2\gamma^2 L^2V_{\max}^2}{(1-\beta)^2} \right)\sum_{t=0}^T \mathbb E\|\nabla f(\bm{x}_t)\|^2_{V}\\
 \leq & \frac{2\mathbb E f(\bm{x}_{0}) - 2\mathbb Ef(\bm{x}^*)}{\gamma}   + \frac{6\gamma^2L^2\epsilon^2V_{\max}^3 T}{(1-\beta)^2}  +  \frac{L\gamma V_{\max}\sigma^2T}{n} + \frac{2\gamma^2L^2\sigma^2 V_{\max}^2 T}{n(1-\beta)^2},
\end{align*}
Since $V_{\max} = \frac{1}{\sqrt{v_{\min}}}$, therefore the equation above becomes
\begin{align*}
  &\left(1-\frac{\gamma L}{v_{\min}} - \frac{2\gamma^2 L^2}{(1-\beta)^2v_{\min}^2} \right)\sum_{t=0}^T \mathbb E\|\nabla f(\bm{x}_t)\|^2_{V}\\
   \leq & \frac{2\mathbb E f(\bm{x}_{0}) - 2\mathbb Ef(\bm{x}^*)}{\gamma}   + \frac{6\gamma^2L^2\epsilon^2 T}{(1-\beta)^2v_{\min}^3}  +  \frac{L\gamma \sigma^2T}{nv_{\min}} + \frac{2\gamma^2L^2\sigma^2  T}{n(1-\beta)^2v_{\min}^2},
\end{align*}
\end{proof}

\section{Proof to Corollary~\ref{coro:global}}
\begin{proof}
By choosing $\gamma =\frac{1-\beta}{4LV_{\max} + \sigma\sqrt{\frac{T}{n}} + T^{^{\frac{1}{3}}}\epsilon^{^{\frac{2}{3}}} }   $, we can guarantee that
\begin{align*}
1- \gamma L - \frac{2\gamma^2 L^2V_{\max}^2}{(1-\beta)^2} \geq & \frac{1}{2}.
\end{align*}
So \eqref{main:theo:eq} leads to
\begin{align*}
\sum_{t=0}^T \mathbb E\|\nabla f(\bm{x}_t)\|^2_V \leq & \frac{2\left(\mathbb E f(\bm{y}_{0}) - f(\bm{y}^*) \right)}{(1-\beta)}\left(4LV_{\max} + \sigma\sqrt{\frac{T}{n}} + T^{^{\frac{1}{3}}}\epsilon^{^{\frac{2}{3}}} \right)\\
& + \left((1-\beta)L\sqrt{T}+ 2L^2V_{\max}^2\right)\frac{\sigma}{\sqrt{n}} +   6L^2\epsilon^{\frac{2}{3}}T^{\frac{1}{3}}V_{\max}^3\\
\frac{1}{T}\sum_{t=0}^T \mathbb E\|\nabla f(\bm{x}_t)\|^2_V \leq & \frac{2\left(\mathbb E f(\bm{y}_{0}) - f(\bm{y}^*) \right)}{(1-\beta)}\left(\frac{4LV_{\max}}{T} + \frac{\sigma}{\sqrt{nT}} + T^{^{-\frac{2}{3}}}\epsilon^{^{\frac{2}{3}}} \right)\\
& + \left((1-\beta)L+ \frac{2L^2V_{\max}^2}{\sqrt{T}}\right)\frac{\sigma}{\sqrt{nT}} +   6L^2\epsilon^{\frac{2}{3}}T^{-\frac{2}{3}}V_{\max}^3.
\end{align*}
Treating $f(\bm{y}_1) - f^*$, $\beta$ and $L$ as constants, from the inequality above we get
\begin{align*}
\frac{1}{T}\sum_{t=0}^T \mathbb E\|\nabla f(\bm{x}_t)\|^2 \lesssim \frac{\sigma}{\sqrt{nT}} + \frac{\epsilon^{\frac{2}{3}}}{T^{\frac{2}{3}}} + \frac{1}{T}.
\end{align*}
It completes the proof.
\end{proof}

\end{document}